\documentclass[lettersize,journal]{IEEEtran}
\usepackage{amsmath,amssymb,amsfonts}
\usepackage{algorithmic}
\usepackage{graphicx}
\usepackage{textcomp}

\usepackage{amsthm,amssymb}
\usepackage{mathrsfs}
\usepackage{enumerate}
\usepackage{float}
\usepackage{color}
\usepackage{bm}
\usepackage{cite}
\usepackage{subcaption}
\usepackage{caption}
\usepackage{tikz}
\usetikzlibrary{arrows.meta, decorations.pathmorphing}
\usetikzlibrary{decorations.pathmorphing, patterns, arrows.meta}
\usetikzlibrary{decorations.pathmorphing, patterns, calc}
\usepackage{hyperref}

\makeatletter

\newcommand{\Rmnum}[1]{\expandafter\@slowromancap\romannumeral #1@}

\newtheorem{assumption}{Assumption}
\newtheorem{lemma}{Lemma}
\newtheorem{remark}{Remark}
\newtheorem{theorem}{Theorem}

\newtheorem{corollary}{Corollary}

\newcommand{\cond}{{\rm cond }}

\newcommand{\revise}[1]{{#1}}
\newcommand{\modify}[1]{{#1}}

\newcommand{\syl}[1]{\textcolor{black}{#1}}
\newcommand{\sun}[1]{\textcolor{black}{#1}}

\newcommand{\response}[1]{\textcolor{black}{#1}}

\newcommand{\ssun}[1]{\textcolor{blue}{#1}}

\def\BibTeX{{\rm B\kern-.05em{\sc i\kern-.025em b}\kern-.08em
		T\kern-.1667em\lower.7ex\hbox{E}\kern-.125emX}}
\usepackage{balance}
\begin{document}
	\title{Finite Sample Analysis of MIMO Systems Identification}
	\author{Shuai Sun, Jiayun Li, and Yilin Mo*, \IEEEmembership{Member, IEEE}
		\thanks{This paper is supported by the BNRist project (No. BNR2024TD03003).}
		\thanks{The authors are with Department of Automation, BNRist, Tsinghua University. Emails: \{suns19, lijiayun22\}@mails.tsinghua.edu.cn, ylmo@tsinghua.edu.cn.}
		\thanks{*: Corresponding author.}}
	
	\markboth{IEEE Transactions on Signal Processing}%
	{Finite \syl{Sample} Analysis of MIMO Systems Identification}
	
	\maketitle
	
	\begin{abstract}
		This paper is concerned with the finite sample identification performance of an $n$ dimensional discrete-time Multiple-Input Multiple-Output (MIMO) Linear Time-Invariant system, with $ p $ inputs and $ m $ outputs. We prove that the widely-used Ho-Kalman algorithm and Multivariable Output Error State Space (MOESP) algorithm are ill-conditioned for MIMO systems when $ n/m $ or $ n/p $ is large. Moreover, by analyzing the Cram\'er–Rao bound, we derive a fundamental limit for identifying the real and stable (or marginally stable) poles of MIMO systems and prove that the sample complexity for \textit{any} unbiased pole estimation algorithm to reach a certain level of accuracy explodes superpolynomially with respect to $ n/(pm) $. Numerical results are provided to illustrate the ill-conditionedness of Ho-Kalman algorithm and MOESP algorithm as well as the fundamental limit on identification. 
	\end{abstract}
	
	\begin{IEEEkeywords}
		System Identification, Ho-Kalman Algorithm, MOESP algorithm, Cram\'er–Rao bound, Fisher Information Matrix
	\end{IEEEkeywords}

\section{Introduction}
Linear Time-Invariant (LTI) systems are an important class of models with many applications in control, signal processing, communication and other engineering fields~\cite{ljung1986system}. Classical results on LTI system identification have mainly focused on guaranteeing asymptotic convergence properties for specific estimation schemes~\cite{goodwin1977dynamic}, such as ordinary least-squares method (OLS)~\cite{ljung1976consistency} and subspace identification methods (SIM)~\cite{deistler1995consistency}. Recently, there has been an increasing interest in finite sample complexity and non-asymptotic analysis~\cite{zheng2020non}, and \modify{significant progress~\cite{campi2002finite,vidyasagar2008learning,hardt2018gradient,care2017finite,simchowitz2019learning,fattahi2019learning,sarkar2019near,tsiamis2019finite,tsiamis2022online,wang2022large, dean2020sample, oymak2021revisiting,jedra2022finite} has been made}.

Depending on whether state measurements are directly available, the systems under consideration can be broadly categorized into two classes. For \textit{fully observed} LTI systems, where the state can be accurately measured, recent works~\cite{simchowitz2018learning,faradonbeh2018finite,sarkar2019near,dean2020sample} derive upper bounds on identification error of the OLS estimator from a finite number of sample trajectories. However, these error bounds usually depend on the true system parameters, which are unknown during the identification process. Dean et al.~\cite{dean2020sample} further provide a data-dependent upper bound on the identification error. Besides OLS, Wagenmaker and Jamieson ~\cite{wagenmaker2020active} propose an active learning algorithm for identification and derive an error bound. A summary on finite \syl{sample} identification of a fully observed system can be found in~\cite{matni2019tutorial}. Notice that the above works mainly focus on the upper bound of finite \syl{sample} identification error for specific identification algorithms such as OLS. On the other hand, Jedra and Proutiere~\cite{jedra2019sample} establish lower bounds on sample complexity for a class of locally-stable algorithms in the Probably Approximately Correct (PAC) framework.
Tsiamis and Pappas~\cite{tsiamis2021linear} reveal that one can craft a fully observed linear system with non-isotropic noise, where the worst-case sample complexity grows exponentially with the system dimension in the PAC framework, independent of identification algorithm used. However, their result is limited to under-actuated and under-excited linear systems with specific structures.

The identification problem of \textit{partially observed} LTI systems is considerably more challenging due to the fact that the state cannot be accurately obtained~\cite{zheng2021sample}. Some recent works~\cite{zheng2020non,oymak2021revisiting,simchowitz2019learning} consider a two-step approach to estimate the state-space model from finite input/output sample trajectories: First step, OLS is used to recover Markov parameters from finite input/output data, which in turn allows recovering a balanced state-space realization via the celebrated Ho-Kalman algorithm. As a result, the upper bounds on the finite \syl{sample} identification error of Markov parameters as well as the system matrices can be derived.
\revise{Another line of research~\cite{chiuso2004ill,hachicha2014n4sid,ikeda2015estimation} focuses on estimating state-space models using the MOESP algorithm, which is a well-known subspace identification method consisting of estimating the (extended) observability matrix followed by the system matrices.}
However, it has been reported in the literature that the subspace method may be ill-conditioned, e.g., see~\cite{chiuso2004ill, hachicha2014n4sid}.
In addition, other learning-based identification methods, such as reinforcement learning method\cite{lale2020logarithmic} and gradient descent algorithm~\cite{hardt2018gradient}, have also been proposed to solve the identification problem.

Another relevant research field is the problem of detecting the number of complex exponentials and estimating their parameters $c_i, \lambda_i$ from noisy signals in the form of $y_k = \sum c_i \lambda_i^k + \text{noise}$, \ssun{commonly known as the harmonic retrieval problem (HRP). The HRP has} widespread applications in digital communications, audio processing, radar systems, and other fields~\cite{kay1981spectrum,percival1993spectral,stoica2005spectral}. In recent decades, a large number of numerical methods have emerged, such as maximum likelihood-based techniques~\cite{lang1980frequency}, Minimum Variance Distortionless Response (MVDR) ~\cite{capon1969high}, SParse Iterative Covariance-based Estimation (SPICE)~\cite{stoica2010new}, optimization-based techniques~\cite{hayes2023sinusoidal}, neural network-based approaches~\cite{pan2021deep}, and information measures approaches~\cite{percival1993spectral,stoica2005spectral}. Alternatively, \ssun{the HRP can be formulated as  a partially-observed LTI system identification problem~\cite{kung1983state,shi1994harmonic}} and solved using Ho-Kalman or MOESP algorithms.

In most of the above research, bounds on the finite \syl{sample} identification error of specific algorithms, such as Ho-Kalman algorithm and MOESP algorithm, have been derived, these bounds usually depend on the true system parameters or input/output trajectories. Hence, it is only possible to prove that the identification problem for a specific system is ill-conditioned. On the other hand, in this paper we provide ``\emph{uniform}'' bound of identification error for the Ho-Kalman and MOESP algorithm. As a result, we are able to prove that the two algorithms are ill-conditioned for MIMO systems whenever $n/m$ or $n/p$ is large regardless of system parameters, where $n,\,p,\,m$ are the dimensions of the state, input and output respectively. Moreover, we show that the ill-conditionedness is not caused by a specific algorithm design, but fundamentally rooted in the identification problem itself. To this end, we derive a lower bound on the identification error of the stable (or marginally stable) and real poles of MIMO systems for \textit{any} unbiased estimator and further analyze the sample complexity for \textit{any} unbiased pole estimation algorithm to reach a certain level of accuracy.

The main contribution of this paper is as follows:
\begin{enumerate}
	\item We prove that the widely-used Ho-Kalman algorithm and MOESP algorithm are ill-conditioned for MIMO systems when $n/m$ or $n/p$ is large.
	\item We \modify{derive a lower bound on the identification error of stable (or marginally stable) and real poles of MIMO systems by analyzing the Fisher Information Matrix used in Cram\'er–Rao bound}.
	\item \modify{We reveal that the sample complexity on the stable (or marginally stable) and real poles of MIMO systems using \textit{any} unbiased estimation algorithm to reach a certain level of accuracy explodes superpolynomially with respect to $n/(pm)$}.
\end{enumerate}
Previous versions~\cite{ascc2022,icca2022,cdc2022} of this work focus on \syl{the analysis of identification limits for} Single-Input Single-Output (SISO) systems, where in this paper we extend the results to general MIMO systems identification.

This paper is organized as follows: Section~\ref{sec:classical algorithms} formulates the problem and evaluates the finite \syl{sample} performance of the widely-used Ho-Kalman algorithm and MOESP algorithm.
\modify{Section~\ref{sec:fisher information matrix} derives the Fisher Information matrix of unknown system poles, which is used in Section~\ref{sec:sample complexity} to lower bound the identification error via Cram\'er–Rao bound and analyze the sample complexity}. Section \ref{sec:numerical results} provides numerical simulations and finally Section \ref{sec:conclusion} concludes this paper.

\textbf{Notations}: 
$\mathbf{0}$ is an all-zero matrix of proper dimensions. For any $ x \in \mathbb{R} $, $ \lfloor x \rfloor $ denotes the largest integer not exceeding $ x $, and $ \lceil x \rceil $ denotes the smallest integer not less than $ x $.
The Frobenius norm is denoted by $ \|A\|_{\rm F} = \sqrt{{\rm tr}\left(A^{\rm H}A\right)} $, and $ \|A\| $ is the spectral norm of $A$, i.e., its largest singular value $ \sigma_{\max}(A) $. 
$ \sigma_{j}(A) $ denotes the $ j $-th largest singular value of $ A $, and $ \sigma_{\min}(A) $ denotes the smallest non-zero singular value of $ A $. 
$ \lambda(A) $ denotes the spectrum of a square matrix $ A $.
The Moore-Penrose inverse of matrix $ A $ is denoted by $ A^\dagger $. $ {\rm cond}(A) = \|A\|\|A^\dagger\| $ denotes the condition number of $ A $.
$ I_n $ denotes the $ n \times n $ identity matrix. 
$ \otimes $ is the Kronecker product.
The matrix inequality $ A \succeq B$ implies that matrix $ A -B $ is positive semi-definite. 
Multivariate Gaussian distribution with mean $ \mu $ and covariance $ \Sigma $ is denoted by $ \mathcal{N}(\mu,\Sigma) $. 
\ssun{The big-$\mathcal{O}$ notation $\mathcal{O}\{f(n)\}$ represents a function that grows at most as fast as $f(n)$.
The big-$\Omega$ notation $\Omega\{f(n)\}$ represents a function that grows at least as fast as $f(n)$.
The big-$\Theta$ notation $\Theta\{f(n)\}$ represents a function that grows as fast as $f(n)$.}

\section{Finite Time Performance Analysis of Classical Identification Algorithms}\label{sec:classical algorithms}
In this section, we formulate the MIMO systems identification problem and evaluate the finite \syl{sample} performance of the widely-used Ho-Kalman algorithm and MOESP algorithm. We consider the identification problem of an observable and controllable\footnote{If the system is not observable or controllable, we could instead focus on the identification of the observable and controllable part of the system.} LTI system evolving according to
\begin{equation}\label{linear_system0}
	\begin{aligned}
		x_{k+1} &= A x_{k}+ B u_{k} , \\
		y_{k} &= 
		C x_{k}+ v_{k},
	\end{aligned}
\end{equation}
\modify{based on finite input/output sample data}, where $x_k \in \mathbb{R}^n$, $u_k \in \mathbb{R}^p$, $y_k \in \mathbb{R}^m$ are the system state, the input and the output, respectively, and $v_k \sim \mathcal{N}(\mathbf{0},\mathcal{R})$ with $\mathcal R>0$ is the \response{i.i.d.} measurement noise. $ A,B,C $ are \textbf{unknown} matrices with appropriate dimensions. For convenience, {define} 
\[
\overline{\bm{\delta}} \triangleq \left(\max_{i,j}|b_{ij}| \right)\left(\max_{i,j}|c_{ij}|\right).
\]

\begin{assumption}\label{assumption1}
	\ssun{Matrix $A$ is diagonalizable, and its eigenvalues, which are the poles of the system~\eqref{linear_system0}, are distinct and real.} The system is both observable and controllable.
\end{assumption}

\begin{remark}
	\response{ It is worth noticing that the observability and controllability condition implies that the state-space representation is minimal. Furthermore, since there are infinitely many state-space models similar to system~\eqref{linear_system0}, which share the same input-output relationship, Assumption~\ref{assumption1} is necessary to make the identification problem well-defined.}
	
	It is also worth noticing that Assumption~\ref{assumption1} can be extended to include systems containing complex or non-single poles, by focusing on the subsystem consisting of distinct real poles. To be specific, we could use a different realization of \eqref{linear_system0} with the following form, without changing the input-output relationship of the original system:
	\begin{displaymath}
		\begin{aligned}
			x_{k+1} &= \begin{bmatrix}
				A_1&\\
				&A_2
			\end{bmatrix}
			x_{k}+ 
			\begin{bmatrix}
				B_1\\
				B_2
			\end{bmatrix} u_{k}, \\
			y_{k} &= 
			\begin{bmatrix}
				C_1&C_2	
			\end{bmatrix} x_{k}+ v_{k},
		\end{aligned}
	\end{displaymath}
	where $A_1$ contains the single real poles and $A_2$ contains the rest of the poles. 
	
	One could assume that the $A_2, B_2, C_2$ matrices are provided by an oracle. As a result, the identification of $(A_1, B_1, C_1)$ can be carried out on the following system:
	\begin{displaymath}
		\begin{aligned}
			\tilde x_{k+1} &= A_1 \tilde x_{k}+B_1u_k, \\
			\tilde y_{k} &= C_1	x_{k}+ v_{k} = y_k - C_2(zI-A_2)^{-1}B_2 u_k.
		\end{aligned}
	\end{displaymath}
	However, as will be shown later, even with this additional information, the identification of $A_1, B_1, C_1$ is still ill-conditioned.
\end{remark}

\begin{remark}
	\ssun{Notice that the problem of estimating parameters $c_i, \lambda_i$ from the noisy signal $y_k = \sum_{i=1}^n c_i \lambda_i^k + \text{noise}$ is known as the harmonic retrieval problem (HRP). As is well-known~\cite{kung1983state,shi1994harmonic}, the HRP can be cast as the identification problem of the following SISO LTI system}:
	\begin{equation}\label{eq:noisysignalest}
		\begin{aligned}
			x_{k+1} &= {\rm diag}(\lambda_1,\ldots,\lambda_n) x_{k} + \mathbf 1 u_k, \\
			y_k &= \begin{bmatrix} c_1&\cdots&c_n\end{bmatrix}x_k + v_k,
		\end{aligned}
	\end{equation}
	where $\mathbf 1$ is an all $1$ vector of proper dimension. It is easy to see that the impulse response of the LTI system is precisely the noisy signal $y_k = \sum_{i=1}^n c_i \lambda_i^k + v_k$. 
\end{remark}

Before continuing on, we shall state several results regarding the (extended) observability, controllability and the Hankel matrices $O$, $Q$ and $H$ of the system \eqref{linear_system0}, which are defined as:
\begin{align}
	O &\triangleq \begin{bmatrix}
		C^\top & (CA)^\top & \cdots & (CA^{K_1-1})^\top
	\end{bmatrix}^\top, \\
	Q &\triangleq \begin{bmatrix}
		B & AB& \cdots & A^{K_2-1}B
	\end{bmatrix},\\
	H &\triangleq \begin{bmatrix}
		CB & CAB &\cdots & CA^{K_2-1}B\\
		CAB & CA^2B &\cdots & CA^{K_2}B\\
		\vdots &  \vdots &\ddots &\vdots \\
		CA^{K_1-1}B & CA^{K_1}B &\cdots  &CA^{K_1+K_2-2}B
	\end{bmatrix} = OQ,
\end{align}
where $K_1$ and $K_2$ are large enough such that $O$ and $Q$ are full column and row rank respectively.  

The condition number of  matrices $ O $ and  $ Q $ and the $n$-th largest singular value of the Hankel matrix $ H $ of the system \eqref{linear_system0} can be characterized by the following Lemma, whose proof is reported in the Appendix: 
\begin{lemma}\label{Ho-Kalman is ill-conditioned}
	Suppose that the matrices $O$ and $Q$ are full column and row rank respectively, then the condition numbers of $ O $ and $ Q $  of the system \eqref{linear_system0} satisfy
	\begin{equation}\label{O and Q}
		{\rm cond}(O) \geq \frac{1}{4}\rho^{\frac{\left\lfloor \frac{n-1}{2m} \right\rfloor}{\log (2mK_1)}}, \quad
		{\rm cond}(Q ) \geq \frac{1}{4}\rho^{\frac{\left\lfloor \frac{n-1}{2p} \right\rfloor}{\log (2pK_2)}},
	\end{equation}
	where  $ \rho \triangleq e^{\frac{\pi^2}{4}} \approx 11.79$.   
	
	Moreover, if the system \eqref{linear_system0} is stable (or marginally stable), then
	the $n$-th largest singular value of $ H $ satisfies the following inequality
	\begin{equation}\label{L}
		\sigma_{n}(H) \leq 2\overline{\bm{\delta}} nK\sqrt{pm} \rho^{-\max \left\{
			\frac{\left\lfloor \frac{n-1}{2m} \right\rfloor}{\log (2mK_1)}, \frac{\left\lfloor \frac{n-1}{2p} \right\rfloor}{\log (2pK_2)}	
			\right\} },
	\end{equation}
	where $K=K_1+K_2$ and $\rho=e^{\frac{\pi^2}{4}}$. 
\end{lemma}

In the following two subsections, we shall apply Lemma~\ref{Ho-Kalman is ill-conditioned} to analyze the finite \syl{sample} performance of the Ho-Kalman algorithm and MOESP algorithm.

\subsection{Analysis of Ho-Kalman Algorithm}
The Ho-Kalman algorithm~\cite{ho1966effective, oymak2021revisiting} is widely used to derive a state-space realization of LTI systems. \ssun{Before introducing the Ho-Kalman algorithm, we first need to define the matrix $G$, consisting of the first $K$ Markov parameter matrices:
\[
	G \triangleq \begin{bmatrix}
		\mathbf{0}_{m \times p} & CB & CAB & \cdots & CA^{K-2}B
	\end{bmatrix} \in \mathbb{R}^{m \times Kp},
\]
which can be estimated using the following least-squares procedure outlined in~\cite{oymak2021revisiting}. }

\ssun{Given the single-trajectory $\{(u_k,y_k)\mid 0 \leq k \leq T-1\}$, we generate $\overline{T}$ subsequences of length $K$, where $T = K + \overline{T} -1$ and $\overline{T} \geq 1$. Define the stacked input $\overline{u}_k \in \mathbb{R}^{Kp}$ as
\[
\overline{u}_k \triangleq  \begin{bmatrix}
	u_{k}^\top &u_{k-1}^\top & \cdots & u_{k-K+1}^\top
\end{bmatrix}^\top.
\]
Let the output matrix \( Y \) and input matrix \( U \) be defined as
\begin{align}\notag
	Y &\triangleq \begin{bmatrix}
		y_{K-1} & y_{K} & \cdots & y_{T-1}
	\end{bmatrix}^\top \in \mathbb{R}^{\overline{T} \times m}, \\\notag
	U &\triangleq \begin{bmatrix}
		\overline{u}_{K-1} & \overline{u}_{K} & \cdots &\overline{u}_{T-1}
	\end{bmatrix}^\top \in \mathbb{R}^{\overline{T} \times Kp}.
\end{align}
The estimation of $G$ can be derived from solving the least-squares problem
\[
\widehat{G} = \arg \min_{X \in \mathbb{R}^{m \times K p}} \| Y - U X^\top \|_{\rm F}^2,
\]
with closed-form solution \( \widehat{G} = (U^\dagger Y)^\top \), where \( U^\dagger = (U^\top U)^{-1} U^\top \) is the left pseudo-inverse of the matrix \( U \).}

\ssun{The Ho-Kalman algorithm is applied to the estimated Markov parameter matrix\footnote{The hat operator $\left(\widehat{\cdot}\right)$ is used to denote an estimated value of the corresponding matrix.} $\widehat{G}$, in accordance with the procedure described in~\cite{oymak2021revisiting}. Its main steps are summarized below.}

\textbf{Step 1}: \ssun{Construct the estimated Hankel matrix $\widehat{H}$ from $\widehat{G}$. Let $ \widehat{H}^+ $ and $ \widehat{H}^- $ be the sub-matrices of $ \widehat{H} $ discarding the left-most and right-most $ mK_1 \times p $ block respectively. Denote $\widehat{L}$ as the best rank-$n$ approximation of $ \widehat{H}^{-}$.}

\textbf{Step 2}: Compute the singular value decomposition (SVD) of the matrix 
\[
\widehat{L} = \mathcal{U}_1 \Sigma_1 \mathcal{V}^{\rm H}_1,
\]
where $\Sigma_1$ is a diagonal matrix, with its elements being the first $n$ largest singular values of matrix $\widehat{L}$, sorted in descending order, and $\mathcal{U}_1$ and $\mathcal{V}_1$ are matrices composed of corresponding left and right singular vectors, respectively. Indeed, $\mathcal{U}_1\Sigma_1\mathcal{V}_1^{\rm H}$ is the best rank-$n$ approximation (or truncated SVD) of $\widehat{H}^{-}$. 

\textbf{Step 3}: Estimate the (extended) observability and controllability matrices $O$ and $Q$ from the SVD of matrix $\widehat{L}$:
\[
\widehat{O} = \mathcal{U}_1 \Sigma_1^{\frac{1}{2}}, \quad \widehat{Q} = \Sigma_1^{\frac{1}{2}}\mathcal{V}_1^{\rm H}.
\]

\textbf{Step 4}: Construct estimates of the state-space matrices as follows
\[
\widehat{C} = \widehat{O}(1:m,:), \quad \widehat{B} = \widehat{Q}(:,1:p), \quad
\widehat{A} = \widehat{O}^\dagger \widehat{H}^{+} \widehat{Q}^\dagger,
\]
where $\widehat{O}(1:m,:)$ is the top-most $ m \times n $ block of $ \widehat{O}$, and $\widehat{Q}(:,1:p)$ is the left-most $ n \times p $ block of $ \widehat{Q}$.

 \ssun{Let $ H^+ $ and $ H^- $ be the sub-matrices of $ H $ discarding the left-most and right-most $ mK_1 \times p $ block respectively. Denote $L$ as the best rank-$n$ approximation of $H^{-}$. Applying the same procedure to the true Markov matrix $G$ results in a true balanced state-space realization $(\overline{A},\overline{B},\overline{C})$ of the system~\eqref{linear_system0}}.

For the Ho-Kalman algorithm above, Oymak and Ozay~\cite{oymak2021revisiting} \syl{establish an end-to-end estimation guarantee on the state-space matrices using a single input/output sample trajectory. This is formalized in the following theorem, which depends on the true parameters of the system}.
\begin{theorem} [\cite{oymak2021revisiting}] \label{ho-kalman algorithm}
	\syl{For the system~\eqref{linear_system0}, consider the setups of Theorem 3.1, 5.2 and 5.3 in~\cite{oymak2021revisiting}.
		Let $T_0 = Kq\log^2(Kq)$, where $q = p+m +n$.
		Suppose\footnote{Since $ L $ is the best rank-$n$ approximation of $H^{-}$, this means that $\sigma_{n}(H^-) = \sigma_{\min}(L)$, thus here we can use $ \sigma_{n}(H^-) $ instead of $\sigma_{\min}(L)$ in \cite{oymak2021revisiting}.} $\sigma_{n}(H^-) >0$ and perturbation obeys
		\begin{equation}\label{perturbation_condition}
			\|L-\widehat{L}\| \leq \frac{ \sigma_{n}(H^-)}{2},
		\end{equation}
		and sample size parameter $T$ obeys
		\begin{equation}
			\frac{T}{\log^2 (Tq)} \sim \Omega \left\{\frac{KT_0}{\sigma_{n}^2(H^-)} \right\}.
		\end{equation}
		Then, with high probability (same as Theorem 3.1 in~\cite{oymak2021revisiting}), there exists a unitary matrix $\mathcal{U} \in \mathbb{R}^{n \times n}$, such that
		\begin{multline}\label{upper_bound_1}
			\max \left\{\|\overline{C}-\widehat{C} \mathcal{U}\|_{\rm F}, \|O-\widehat{O} \mathcal{U}\|_{\rm F}, \|\overline{B}-\mathcal{U}^{\rm H} \widehat{B} \|_{\rm F}, \right. \\ \left.
			\| Q-\mathcal{U}^{\rm H} \widehat{Q} \|_{\rm F} \right\}
			\leq \frac{\sqrt{\mathcal{C}nK}\log(Tq)}{\sqrt{\sigma_{n}(H^-)}} \sqrt{\frac{T_0}{T}},
		\end{multline}
		where $\mathcal{C}$ is a constant.
		Furthermore, hidden state matrices $\widehat{A}, \overline{A}$ satisfy
		\begin{equation}\label{upper_bound_2}
			\|\overline{A}-\mathcal{U}^{\rm H} \widehat{A} \mathcal{U}\|_{\rm F} \leq \frac{\mathcal{C}\sqrt{nK}\log(Tq)\|H\|}{\sigma_{n}^2(H^-)}  \sqrt{\frac{T_0}{T}}.
	\end{equation}}
\end{theorem}

\response{
	According to the above theorem, if the system is stable (or marginally stable), then to ensure a reasonable system identification performance, the error $\|L - \widehat L\|$ should be comparable with respect to $\sigma_n(H^-)$. }

\response{
	Using essentially the same argument of Lemma~\ref{Ho-Kalman is ill-conditioned}, one can show that, if the system \eqref{linear_system0} is stable (or marginally stable), then the following inequality holds:
	\begin{equation}\label{L2}
		\sigma_{n}(\ssun{H^-}) \leq 2\overline{\bm{\delta}} nK\sqrt{pm} \rho^{-\max \left\{
			\frac{\left\lfloor \frac{n-1}{2m} \right\rfloor}{\log (2m(K_1-1))}, \frac{\left\lfloor \frac{n-1}{2p} \right\rfloor}{\log (2pK_2)}	
			\right\} },
	\end{equation}
	where $K=K_1+K_2$ and $\rho=e^{\pi^2/4}\approx 11.79$. As a result, it can be seen that $\sigma(\ssun{H^-})$ decays superpolynomially with respect to the largest number between $n/m$ or $n/p$. 
}

\syl{Oymak and Ozay in~\cite{oymak2021revisiting} point out that, ignoring logarithmic factors, when $T_0 \sim \Theta \{Kq\}$, to achieve constant estimation error with high probability, the sample size parameter $T$ need to satisfy $T \sim \Omega \{K^2 qn\}$. However, based on Theorem~\ref{ho-kalman algorithm} and~\eqref{L2}, it can be concluded that, in fact, achieving constant estimation error accuracy \ssun{with high probability for matrices} $\widehat{A}, \overline{A}$ requires that the sample size parameter $T$ satisfy 
	\begin{equation}
		T \sim \Omega \left\{ \frac{q}{npm} \varrho^{\max \left\{\frac{\left\lfloor \frac{n-1}{2m} \right\rfloor}{\log (2m(K_1-1))}, \frac{\left\lfloor \frac{n-1}{2p} \right\rfloor}{\log (2pK_2)}\right\}}  \right\},
	\end{equation}
	where $\varrho \triangleq e^{\pi^2} \approx 19333.69$ and $q = p+m+n$. Note that in~\cite{oymak2021revisiting}, the length of input/output sample trajectory satisfies that $T_{\rm total} = T + K -1$. This implies that ignoring logarithmic factors, achieving constant estimation error accuracy \ssun{with high probability for matrices} $\widehat{A}, \overline{A}$ requires the sample length $T_{\rm total}$ to grow superpolynomially with respect to \ssun{$\max\{n/p, n/m\}$.}
}

On the other hand, if $N$ independent sample trajectories $ \left\{(u_k^{(\ell)},y_{k}^{(\ell)}) \mid \ell = 1,2,\cdots,N; k = 0,1,\cdots, T_{\rm total}  \right\} $ are collected, then the convergence of estimated Markov parameters to the true Markov parameters is governed by the central limit theorem. Hence, the variance of $\|L-\widehat L\|$ is $\mathcal{O}\{1/{N}\}$. \ssun{Based on Theorem V.2 in~\cite{oymak2021revisiting} and~\eqref{L2}, it can be concluded that, in fact, achieving constant estimation error accuracy with high probability for matrices $\widehat{A}, \overline{A}$ requires that the number of trajectories satisfy 
\begin{equation}
	N \sim \Omega \left\{ \frac{1}{pm n^2 T_{\rm total}^2} \varrho^{\max \left\{\frac{\left\lfloor \frac{n-1}{2m} \right\rfloor}{\log (2m(K_1-1))}, \frac{\left\lfloor \frac{n-1}{2p} \right\rfloor}{\log (2pK_2)}\right\}}  \right\},
\end{equation}
which demonstrates that the required number of trajectories $N$ grows superpolynomially with respect to $\max\{n/p, n/m\}$.}

\subsection{Analysis of MOESP Algorithm}
In this subsection, we analyze the finite sample performance of the widely-used MOESP algorithm~\cite{verhaegen1992subspace}. Before introducing the MOESP algorithm, we need to give the definitions of some symbols. First, we define the Hankel matrix of past inputs and future inputs as follows
	\begin{equation}\label{U_p}
		U_p \triangleq \begin{bmatrix}
			u_0 & \cdots & u_{K_2 -1} \\
			\vdots & \ddots & \vdots \\
			u_{K_1 - 1} & \cdots & u_{K_1+K_2-2}
		\end{bmatrix},
	\end{equation}
	\ssun{and
	\begin{equation}\label{U_f}
		U_f \triangleq \begin{bmatrix}
				u_{K_1} & \cdots & u_{K_1 + K_2 -1} \\
				\vdots & \ddots & \vdots \\
				u_{2K_1 - 1} & \cdots & u_{2K_1+K_2-2}
			\end{bmatrix}.
	\end{equation}
	The Hankel matrices of past and future outputs can be defined similarly as $ Y_p $ and $Y_f$, respectively. The Hankel matrix of past data is defined as $ Z_p \triangleq \begin{bmatrix}
		U_p^\top & Y_p^\top
	\end{bmatrix}^\top $. }
			
We briefly summarize the main steps of the MOESP algorithm\footnote{Here we only give the MOESP algorithm for a single sample trajectory. For multiple trajectories, Laurent et al.~\cite{duchesne1995subspace} describe how to execute it.} as follows.

\textbf{Step 1}: \ssun{Construct the Hankel matrices $U_f$, $Y_f$ and $Z_p$}  based on finite input/output sample data.

\textbf{Step 2}: \ssun{Perform LQ decomposition as  follows:
	\begin{equation}	 \begin{bmatrix}
			U_f \\ Z_p \\ Y_f
		\end{bmatrix} 
		= \begin{bmatrix}
			L_{11} & \mathbf{0} &  \mathbf{0}\\
			L_{21} & L_{22} &  \mathbf{0}\\
			L_{32} & L_{32} &  L_{33} 
		\end{bmatrix}
		\begin{bmatrix}
			Q_1^\top \\ Q_2^\top \\ Q_3^\top
		\end{bmatrix}.
\end{equation}}

\textbf{Step 3}: Estimate the (extended) observability matrix $O$ from SVD of \ssun{$\frac{1}{\sqrt{K_2}}L_{32}$}:
\begin{equation}
	\widehat{O} = \mathcal{U}_2 \Sigma_2^{\frac{1}{2}},
\end{equation}
where $\mathcal{U}_2 \Sigma_2 \mathcal{V}_2^{\rm H}$ is the truncated SVD of \ssun{$\frac{1}{\sqrt{K_2}}L_{32}$}. \ssun{Here,} $\Sigma_2$ is a diagonal matrix, with its elements being the first $n$ largest singular values of matrix \ssun{$\frac{1}{\sqrt{K_2}}L_{32}$}, sorted in descending order, and $\mathcal{U}_2$ and $\mathcal{V}_2$ are matrices composed of corresponding left and right singular vectors, respectively.

\textbf{Step 4}: Compute estimates of the state-space matrices as follows
\begin{equation}
	\widehat{C} = \widehat{O}(1:m,:), \quad
	\widehat{A} = \underline{\widehat{O}}^\dagger \overline{\widehat{O}},
\end{equation}
\response{where $\underline{\widehat{O}}$ and $\overline{\widehat{O}}$ are the sub-matrices of $\widehat{O}$ discarding the bottom-most and top-most $m\times n$ block respectively.} Compute least square estimates of matrix $B$ using \sun{equation}\footnote{Due to the large number of linear equations and symbols involved in solving the least squares method of $B$, limited by space constraints, no specific equations and symbols are given here, and more details can be found in~\cite{tangirala2018principles}.} (23.120) in~\cite{tangirala2018principles}.

\ssun{The following theorem provides the numerical perturbation results for the estimation of matrix $A$ in the MOESP algorithm}.

\begin{theorem}\label{moesp}
	\ssun{For the system~\eqref{linear_system0}, let $O$ be the (extended) observability matrix, and $\widehat{O}$ be the estimated (extended) observability matrix obtained from Step 3 of the MOESP algorithm.  Suppose that ${\rm rank} [\underline{O}] = {\rm rank} [\underline{\widehat{O}}] = n$, 
	and that $ \| \Delta_{O}\| < \sigma_{\min}(\underline{O})$, where $\underline{O}$ is the sub-matrix of $\widehat{O}$ obtained by discarding the bottom-most $m \times n$ block, and $\Delta_{O} = \widehat{O} - O$. Then, 
	\begin{equation}
		\| \widehat{A} - \overline{A}
		\| \leq 
		\frac{\cond(\underline{O} )  \| \Delta_{O} \|}{ \| \underline{O} \| -  \cond(\underline{O} ) \| \Delta_{O} \|}
		\left( 1 + \| \overline{A}\|
		\right).
	\end{equation}}
\end{theorem}

\ssun{According to Lemma~\ref{Ho-Kalman is ill-conditioned}, the condition number of $O$ increases at a superpolynomial rate with respect to $n/m$. Combining this with the results from Theorem~\ref{moesp}, it can be concluded that when $n/m$ is large, applying the MOESP algorithm to identify matrix $A$ may result in significant numerical errors. In this sense, the MOESP algorithm can become ill-conditioned. On the other hand, if we perform a perturbation analysis on $L_{32}$, which is  defined in Step 2 of the MOESP algorithm, similar to the perturbation analysis of $L$ in the Ho-Kalman algorithm as presented in Theorem~\ref{ho-kalman algorithm}, we can obtain identification error analysis results for matrices $A$, $B$, and $C$ that are analogous to those in the Ho-Kalman case. Specifically, ignoring logarithmic factors, achieving constant estimation error accuracy with high probability for matrices $\widehat{A}, \overline{A}$ requires the sample length $T_{\rm total} = 2K_1 + K_2 -1$ to grow superpolynomially with respect to $n/m$. }

\begin{remark}
	\ssun{Lemma~\ref{Ho-Kalman is ill-conditioned} establishes that the condition numbers of the observability matrix $O$ and the controllability matrix $Q$ grow at a superpolynomial rate with respect to $n/m$ and $n/p$, respectively, while the smallest singular value of the Hankel matrix $H$ decays at a superpolynomial rate with respect to $\max\{n/m,\, n/p\}$. The numerical ill-conditioning of these matrices explains, from a numerical computation perspective, why the Ho-Kalman and MOESP algorithms tend to exhibit ill-conditioning when $n/m$ or $n/p$ is large. It should be noted that the results presented in Section~\ref{sec:classical algorithms} do not constitute a lower bound on the sample complexity and that a detailed analysis of the sample complexity is provided in Section~\ref{sec:sample complexity}.}
\end{remark}


\section{Analysis of Fisher Information matrix}\label{sec:fisher information matrix}

In the previous section, we prove that the widely-used Ho-Kalman algorithm and MOESP algorithm are ill-conditioned for MIMO systems when $n/m $ or $n/p$ is large.
One may wonder if the ill-conditionedness is due to the specific implementation of the Ho-Kalman algorithm and the MOESP algorithm, or is it fundamentally rooted in the identification problem itself.
In the following, we prove that the identification problem itself is inherently ill-conditioned in the sense that \textit{any} unbiased estimator will have an estimation error covariance that grows superpolynomially with respect to $n/(pm)$. To this end, we derive the Fisher Information matrix and further characterize its smallest eigenvalue, which shall be used to lower bound the sample complexity of \textit{any} unbiased identification algorithm in the next section via Cram\'er-Rao bound.

\subsection{Problem Formulation}
In order to make the identification problem well-defined, here we consider the following observable and controllable LTI system:
\begin{equation}\label{linear_system}
	\begin{aligned}
		x_{k+1} &= \underbrace{\begin{bmatrix}
				\lambda_1 & \\
				& \lambda_2 & \\
				& & \ddots \\
				& & & \lambda_n
		\end{bmatrix}}_{A} x_{k}+ 
		\underbrace{\begin{bmatrix}
				b_{11} & \cdots & b_{1p} \\
				b_{21} & \cdots & b_{2p} \\
				\vdots & \ddots & \vdots \\
				b_{n1} & \cdots & b_{np} \\
		\end{bmatrix}}_B u_{k}, \\
		y_{k} &= 
		\underbrace{\begin{bmatrix}
				c_{11} & c_{12} & \cdots & c_{1n} \\
				\vdots & \vdots & \ddots & \vdots \\
				c_{m1} & c_{m2} & \cdots & c_{mn} \\
		\end{bmatrix}}_{C} x_{k}+ v_{k},
	\end{aligned}
\end{equation}
\revise{which is the diagonal form of the system \eqref{linear_system0}. }

\begin{remark}
	It is worth noticing that there exists infinitely many equivalent state-space models of the system, which are compatible with input and output data and are similar to each other. Hence, To make the identification problem well-conditioned, we choose to identify the diagonal form in this section, since poles play an important role in various controller design methods, e.g., loop-shaping, root locus and so on. We plan to investigate the identifiability of other canonical form, such as the controllability standard form in the future.
	
\end{remark}

We further make the following assumptions besides Assumption~\ref{assumption1}:
\begin{assumption}\label{assumption2}
	\begin{enumerate}
		\item $ A $ is \textbf{unknown}, and \ssun{$ B,C $ are known}.
		\item All eigenvalues of $ A $ are stable (or marginally stable), i.e., all eigenvalues of $ A $ lie on the line segment $ [-1,1] $.
		\item For each trajectory, the initial condition of the state $x_0 = \mathbf{0}$ is known.
		\item Without loss of generality\footnote{One can scale $B$ to make the assumption hold.}, we assume that for each trajectory, the energy of input does not exceed $ 1 $ , i.e, $  \sum_{k=0}^{K-1} \| u_k^{(\ell)} \|^2 \leq 1$, where the superscript $\ell$ here represents the label of sample trajectories.
		\item \ssun{The measurement noise $ v_k $ obeys the standard Gaussian, i.e., $ \mathcal{R} = I_m $.}
		\item The length $K$ of each trajectory is the same. 
	\end{enumerate}
\end{assumption}

\begin{remark}
	Notice that \ssun{the assumption that $B,C$ as well as the initial condition $x_0=\mathbf{0}$ are known essentially makes the identification problem easier}. However, as shall be proved later, even if the identification algorithm has such knowledge, the identification problem is still ill-conditioned when $n/(pm)$ is large.  	
\end{remark}

The main goal of this section is to derive and bound the Fisher Information matrix of unknown poles $ \bm{\lambda} $ for MIMO systems, where $ \bm{\lambda} \triangleq \begin{bmatrix}
	\lambda_1 & \cdots & \lambda_n 
\end{bmatrix}^\top $, based on \textbf{finite} number of \textit{sample trajectories}, i.e., a sequence $ \left\{\left(u_k^{(\ell)},y_{k+1}^{(\ell)}\right) \right\}_{k=0}^{K-1} $ with $ \ell = 1,2,\cdots,N $, where $ N $ is the number of available sample trajectories and $K$ is the length of each trajectory.

\subsection{Fisher Information Matrix}
\modify{In this subsection, we derive the Fisher Information matrix $ \mathbf{I}(\bm{\lambda}) $ of unknown poles $ \bm{\lambda} $ of MIMO systems and further characterize its smallest eigenvalue}.

\begin{theorem}\label{fisher information matrix}
	For the system \eqref{linear_system}, under the condition that Assumption~\ref{assumption1} and \ref{assumption2} are satisfied, given one sample trajectory with length $ K $, 
	the Fisher Information matrix $ \mathbf{I}(\bm{\lambda}) $ of unknown poles $ \bm{\lambda} $ is
	\begin{equation}\label{I = SV}
		\mathbf{I}(\bm{\lambda}) =  V^\top S^\top S V,
	\end{equation}
	where 
	\begin{equation}\label{S}
		S = \begin{bmatrix}
			I_m \otimes u_0^\top &  &     \\
			\vdots & \ddots  &  \\
			I_m \otimes u_{K-1}^\top  & \cdots & 	I_m \otimes u_0^\top
		\end{bmatrix},
	\end{equation}
	and
	\begin{equation}\label{V}
		V = 
		\begin{bmatrix}
			0  & \cdots & 0 \\
			\vdots  & \ddots & \vdots \\
			0  & \cdots & 0\\
			c_{11}b_{11}  &\cdots & c_{1n}b_{n1} \\
			\vdots  & \ddots & \vdots \\
			c_{m1}b_{1p}  &\cdots & c_{mn}b_{np} \\
			\vdots  & \ddots & \vdots \\
			(K-1)c_{11}b_{11}\lambda_1^{K-2}  &\cdots & (K-1)c_{1n}b_{n1}\lambda_n^{K-2} \\
			\vdots  & \ddots & \vdots \\
			(K-1)c_{m1}b_{1p}\lambda_1^{K-2}  &\cdots & (K-1)c_{mn}b_{np}\lambda_n^{K-2} \\
		\end{bmatrix}.
	\end{equation}
	
\end{theorem}

\begin{corollary}\label{fisher information matrix with N sample}
	For the system \eqref{linear_system}, under the condition that Assumption~\ref{assumption1} and \ref{assumption2} are satisfied, given $ N $ sample trajectories with length $ K $, the Fisher Information matrix $ \mathbf{I}(\bm{\lambda}) $ of unknown poles $ \bm{\lambda} $ is
	\begin{equation}
		\mathbf{I}(\bm{\lambda}) =  \sum_{\ell=1}^N V^\top (S^{(\ell)})^\top S^{(\ell)} V,
	\end{equation}
	where the superscript $ \ell $ of the matrix $ S^{(\ell)} $ indicates that it is a matrix formed by the input on the $ \ell $-th sample trajectory in the form of \eqref{S}.
\end{corollary}

We have the following lemma to bound $(S^{(\ell)})^\top S^{(\ell)}$:
\begin{lemma}\label{bound of S}
	For each trajectory, if the energy of the input is no more than $ 1 $, then 
	\begin{equation}\label{SS}
		(S^{(\ell)})^\top S^{(\ell)} \preceq pmK^2 I_{pmK}.
	\end{equation}
	Further, if there are $ N $ sample trajectories, 
	we have
	\begin{equation}
		\mathbf{I}(\bm{\lambda}) \preceq  pmNK^2  V^\top  V,
	\end{equation}
\end{lemma}

\begin{proof}
	According to the definition of $ S^{(\ell)} $, we have
	\begin{equation}
		\|S^{(\ell)}\| \leq \sqrt{mK} \|S^{(\ell)}\|_{\infty}
		= \sqrt{mK}  \sum_{k=0}^{K-1} \|u_k^{(\ell)}\|_1 
		\leq  K\sqrt{pm},
	\end{equation}
	where the second inequality holds because of Cauchy-Schwartz inequality and that the energy of input of each sample trajectory is not more than $ 1 $. Hence,
	\begin{equation}
		(S^{(\ell)})^\top S^{(\ell)} \preceq pmK^2 I_{pmK},
	\end{equation}
	which completes the proof.
\end{proof}

\begin{remark}
	\ssun{
	It is important to note that the assumption in Assumption~\ref{assumption2}, which states that the total energy of the input does not exceed 1, may not be entirely natural. However, the current analysis framework can easily accommodate the case where the total input energy is linearly proportional to the trajectory length. Specifically, if the average power $P$ of the input is bounded, we can modify the total energy constraint from 1 to $PK$, where $K$ is the trajectory length. As a result, Inequality~\eqref{SS} in Lemma~\ref{bound of S} needs to be adjusted to $S^\top S \preceq pmPK^3 I_{pmK}$.
	Since the total energy only affects the sample complexity and other results polynomially, this adjustment does not alter the conclusion regarding the superpolynomial sample size requirement.}
\end{remark}

According to Lemma~\ref{bound of S}, it follows that the smallest eigenvalue of $ \mathbf{I}(\bm{\lambda})  $ satisfies 
\begin{equation}\label{upper bound of min I}
	\lambda_{\min}(\mathbf{I}(\bm{\lambda}))  \leq pmNK^2\sigma_{\min}^2(V).
\end{equation}
Thus, to show that $\mathbf{I}(\lambda)$ is ill-conditioned, we only need to focus on $\sigma_{\min}(V)$.
For the convenience of notation, we \textbf{use $ \bm{\kappa} $ to denote $ \bm{\frac{n}{pm}} $} in the following.

\begin{lemma}\label{singular value of V}
	\sun{Suppose that the matrix $ V $ defined in \eqref{V} is full column rank, then its smallest singular value satisfies that}
	\begin{equation}
		\sigma_{\min}^2(V)\leq 16 n^2 pm \overline{\bm{\delta}}^2 K^3  \rho^{-\frac{\lfloor \bm{\kappa} \rfloor -3}{\log (2K)}},
	\end{equation}
	where $ \rho = e^{\frac{\pi^2}{4}} $.
\end{lemma}
\modify{
	Further, we can get the following theorem describing the characterization of the smallest eigenvalue of the Fisher Information matrix $ \mathbf{I}(\bm{\lambda})$}.

\begin{theorem}\label{worst case}
	For the system \eqref{linear_system}, under the condition that Assumption~\ref{assumption1} and \ref{assumption2} are satisfied, based on $ N $ sample trajectories with length $ K $, \sun{suppose that the Fisher Information matrix $ \mathbf{I}(\bm{\lambda}) $ is full rank,} then its smallest eigenvalue satisfies that
	\begin{equation}\label{worst case upper bound}
		\lambda_{\min}\left(\mathbf{I}(\bm{\lambda})\right) \leq 
		16N n^2(pm)^2 \overline{\bm{\delta}}^2 K^5 \rho^{-\frac{\lfloor \bm{\kappa} \rfloor-3}{\log(2K)}},
	\end{equation}
	where $ \rho = e^{\frac{\pi^2}{4}} $.
\end{theorem}

Theorem \ref{worst case} above reveals that for MIMO systems, the smallest eigenvalue of the Fisher information matrix $\lambda_{\min}\left(\mathbf{I}(\bm{\lambda})\right) $ decreases at a superpolynomial rate with respect to $ \bm{\kappa} $ ($ \bm{\kappa} = \frac{n}{pm} $). In the next section, we shall leverage this result to prove that the system identification problem is ill-conditioned for any unbiased estimator.

\section{Analysis of Sample Complexity}\label{sec:sample complexity}
\modify{In this section, we derive a fundamental limit on identifying unknown poles $ \bm{\lambda} $ of MIMO systems via Cram\'er-Rao bound. Specifically, we prove that \textit{any} unbiased estimator will have an estimation error covariance that grows superpolynomially with respect to $ \bm{\kappa} $, in this sense that the identification problem itself is ill-conditioned. Moreover, we reveal that the sample complexity required to reach a certain level of accuracy for \textit{any} unbiased pole estimation algorithm explodes superpolynomially with respect to $ \bm{\kappa} $.}

\modify{
	\begin{theorem}[\ssun{Section 3.8 in~\cite{kay1993fundamentals}}][Cram\'er–Rao bound]\label{cramer-rao bound}
		It is assumed that the probability density function $ f_{\bm{\theta}}(\bm{x}) $ satisfies the regularity conditions
		\[
		\mathbb{E}_\theta \left[\frac{\partial \log f_{\bm{\theta}}(\bm{x})}{\partial \bm{\theta}} \right] = \mathbf{0}, \text{ for all } \bm{\theta},
		\]
		where the expectation is taken with respect to $ f_{\bm{\theta}}(\bm{x}) $. \textcolor{blue}{Let  $\widehat{\bm{\theta}}$ be an estimate of $ \bm{\theta}$ and denote $\mathbb E_\theta [\widehat{\bm{\theta}}]$ by $\psi( \bm{\theta}) $, then the following inequality holds:
        \begin{equation}
        	{\rm Cov}\left[\widehat{\bm{\theta}}\right] \succeq J\psi^\top \mathbf{I}^{-1}(\bm{\theta}) J\psi,
        \end{equation}
	where $J\psi$ is the Jacobian matrix of $\psi( \bm{\theta})$, i.e., the $ij$-th entry of $J\psi$ is given by $\partial \psi_i( \bm{\theta}))/\partial  \bm{\theta}_j$, and $ \mathbf{I}(\bm{\theta}) $ is the Fisher Information matrix of $ \bm{\theta} $. }
        
        Specifically, $\widehat{ \bm{\theta}}$ is unbiased, i.e., if $\psi( \bm{\theta})=  \bm{\theta}$, then the the covariance matrix satisfies
                \[
                {\rm Cov}[\widehat{\bm{\theta}}] \succeq \mathbf{I}^{-1}(\bm{\theta}).
		\]
	\end{theorem}
}

Based on Theorem \ref{worst case} and Theorem \ref{cramer-rao bound} above, the following theorem related to the sample complexity can be easily obtained.

\begin{theorem}\label{worst case-cramer-rao}
  For the system \eqref{linear_system}, under the condition that Assumption \ref{assumption1} and \ref{assumption2} are satisfied, based on $ N $ sample trajectories with length $ K $, suppose that the Fisher Information matrix $ \mathbf{I}(\bm{\lambda}) $ is full rank. \textcolor{blue}{Let $ \widehat{\bm{\lambda}} $ be an estimator (possibly biased) of unknown system poles $ \bm{\lambda}$. If the Jacobian of $\psi(\bm{\lambda})=\mathbb E_{\bm \lambda} [\widehat{\bm{\lambda}}]$ is uniformly bounded from below, then the largest eigenvalue of its covariance matrix satisfies}
	\begin{equation}\label{worst case-cramer-rao1}
		\begin{aligned}
			&\lambda_{\max}\left({\rm Cov}[\widehat{\bm{\lambda}}] \right)\\ 
			&\sim 	\Omega  \left\{
			\frac{1}{ N   \overline{\bm{\delta}}^2}
			\exp \left[ \frac{\pi^2 \bm{\kappa}}{4 \log K} -2 \log (npm)
			-5\log K 
			\right]
			\right\},
		\end{aligned}
	\end{equation}
	where $ \rho = e^{\frac{\pi^2}{4}} $. 
\end{theorem}
\begin{remark}
  \textcolor{blue}{It is worth noticing that if the estimator $\widehat{\bm{\lambda}}$ is unbiased, then $\psi(\bm{\lambda})=\bm{\lambda}$ and the Jacobian matrix is identity. Therefore, to reduce the bias of the estimator, the Jacobian matrix should be close to identity. Moreover, the second moment of the estimation error is given by $\rm{Cov}[\widehat{\bm{\lambda}}] + (\bm{\lambda}-\psi(\bm{\lambda}))(\bm{\lambda}-\psi(\bm{\lambda}))^\top$, which implies that estimation error is large if the Fisher information matrix is ill-conditioned.
  }
\end{remark}

Now, we analyze the sample complexity for \textit{any} unbiased pole estimation algorithm, i.e., the number of samples needed to reach \modify{a certain level of accuracy}. 
Notice that according to Theorem \ref{worst case-cramer-rao}, the largest eigenvalue of the covariance matrix using \textit{any} unbiased estimator $ \widehat{\bm{\lambda}} $ depends on both the number $ N $ of sample trajectories  and the length $ K $ of each trajectory. In the following we shall consider two extreme cases:

\begin{enumerate}[(I).]
	\item Consider the case where the identification is conducted using \textbf{multiple trajectories with the shortest possible length\footnote{
			\sun{In fact, since the $n$ poles $\{\lambda_1, \cdots, \lambda_n\}$ are all unknown parameters of the system \eqref{linear_system}, the dimension of the output is $m$, and the first output $y_0$ of each sample trajectory contains only the information of the measurement noise, in order to make sense of this identification problem, the length $K$ of the sample trajectory is at least $\lceil n/m \rceil + 1$.} }$K = \lceil n/m\rceil + 1$}. 
	According to Theorem~\ref{worst case-cramer-rao}, \response{
		if the largest eigenvalue of the covariance matrix of unbiased estimator $ \widehat{\bm{\lambda}} $ of unknown poles is no more than a fixed constant $\epsilon > 0$, then the number $ N $ of sample trajectories must satisfy}
	\begin{equation}\label{minimize length worst}
		\small
			N \sim \Omega \left\{\frac{	\exp \left[ 
				\frac{\pi^2 \bm{\kappa}}{4 \log \frac{n}{m}} - 7\log n - 2\log (pm) + 5 \log m
				\right]}{\epsilon \overline{\bm{\delta}}^2}
			\right\}.
	\end{equation}
	\response{
		Note that when $ \bm{\kappa} $ is much larger than $ pm $, 
		we can further obtain that $ N \sim \Omega \left\{\frac{1}{\epsilon \overline{\bm{\delta}}^2} \bm{\kappa}^{-16}11.8^{\bm{\kappa}^{1-\varepsilon} } \right\} $, where $ \varepsilon > 0 $ is any small constant. Hence, the number $ N $ of trajectories required grows superpolynomially with respect to $ \bm{\kappa} $.}
	
	\item Consider the case where \textbf{only a single trajectory} is used. 
	\response{If the largest eigenvalue of the covariance matrix of unbiased estimator $ \widehat{\bm{\lambda}} $ of unknown poles is no more than a fixed constant $\epsilon > 0$, then} the length $ K $ of the trajectory must satisfy
	\begin{equation}\label{one trajectory_worst}
		\small
			K \sim \Omega \left\{
			\frac{	\exp \left[ \left( 5\pi^2  \bm{\kappa} + \log^2 \left( 16 \epsilon (\overline{\bm{\delta}} npm)^2
				\right)
				\right)^{0.5}/10 \right]}{
		\epsilon^{0.1} (\overline{\bm{\delta}} npm)^{0.2}} 
			\right\}.
	\end{equation}
	\response{
		Note that when $ \bm{\kappa} $ is much larger than $ pm $ and $16  \overline{\bm{\delta}}^2 n^2 (pm)^2 \epsilon \geq 1$, 
		the term in the brace of \eqref{one trajectory_worst} can be approximated by $\epsilon^{-0.1} \overline{\bm{\delta}}^{-0.2} \bm{\kappa}^{-0.6}2.02^{\sqrt{\bm{\kappa}} } $.
		Hence,  the length $ K $ of the trajectory required also grows superpolynomially with respect to $ \bm{\kappa}$.
	}
\end{enumerate}
\begin{remark}
	Notice that the lower bound on the sample complexity with a single ``long'' trajectory, albeit still superpolynomial with respect to $ \bm{\kappa} $, is significantly smaller than that of multiple ``short'' trajectories, which suggests that it may be beneficial to use a longer trajectory rather than multiple but shorter trajectories. However, it may be computationally challenging to use a single long trajectory, as both Ho-Kalman and MOESP algorithm requires factorization of matrix whose size is proportional to the length of the trajectory.
\end{remark}

\begin{remark}
	\response{
		It is worth noticing that the ill-conditionedness of the identification problem is caused by the bad numerical conditions of certain matrices (observability, controllability, Hankel and Fisher information matrix) {and} the stochastic noise presented in the sensory data. In essence, our results indicate that even very small amount of noise would greatly affect the identification results, and as a consequence, the sample complexity to get a reasonably accurate system model is very high since the noise level needs to be extremely small.  }
\end{remark}

\section{Numerical Results}\label{sec:numerical results}

\response{
	In this section, we first consider the identification of a discrete LTI system modeling a heat conduction process~\cite{mo2009network} in a planar closed region as
	\begin{equation}\label{heat conduction process}
		\frac{\partial \tau}{\partial t} = \alpha \left( \frac{\partial^2 \tau}{\partial x_1^2} + \frac{\partial^2 \tau}{\partial x_2^2}\right)
	\end{equation}
	with boundary conditions
	\begin{equation}
		\left. \frac{\partial \tau}{\partial x_1} \right|_{t,0,x_2} = 
		\left. \frac{\partial \tau}{\partial x_1} \right|_{t,l,x_2} =
		\left. \frac{\partial \tau}{\partial x_2} \right|_{t,x_1,0} =
		\left. \frac{\partial \tau}{\partial x_2} \right|_{t,x_1,l} =0,
	\end{equation}
	where $x_1$ and $x_2$ represent the coordinates of this planar region, which is a square with slide length $l$, and $\tau(t,x_1,x_2)$ denotes the temperature at time $t$ at location $(x_1,x_2)$ and the parameter $\alpha$ denotes the thermal conductivity. We use the finite difference method to discretize the PDE into a $4 \times 4$ grid with $1$Hz sample frequency. 
}

\begin{figure*}[htbp]
	\centering
	\subcaptionbox{\label{pde_id}}[.26\linewidth]
	{\includegraphics[width = 1\linewidth]{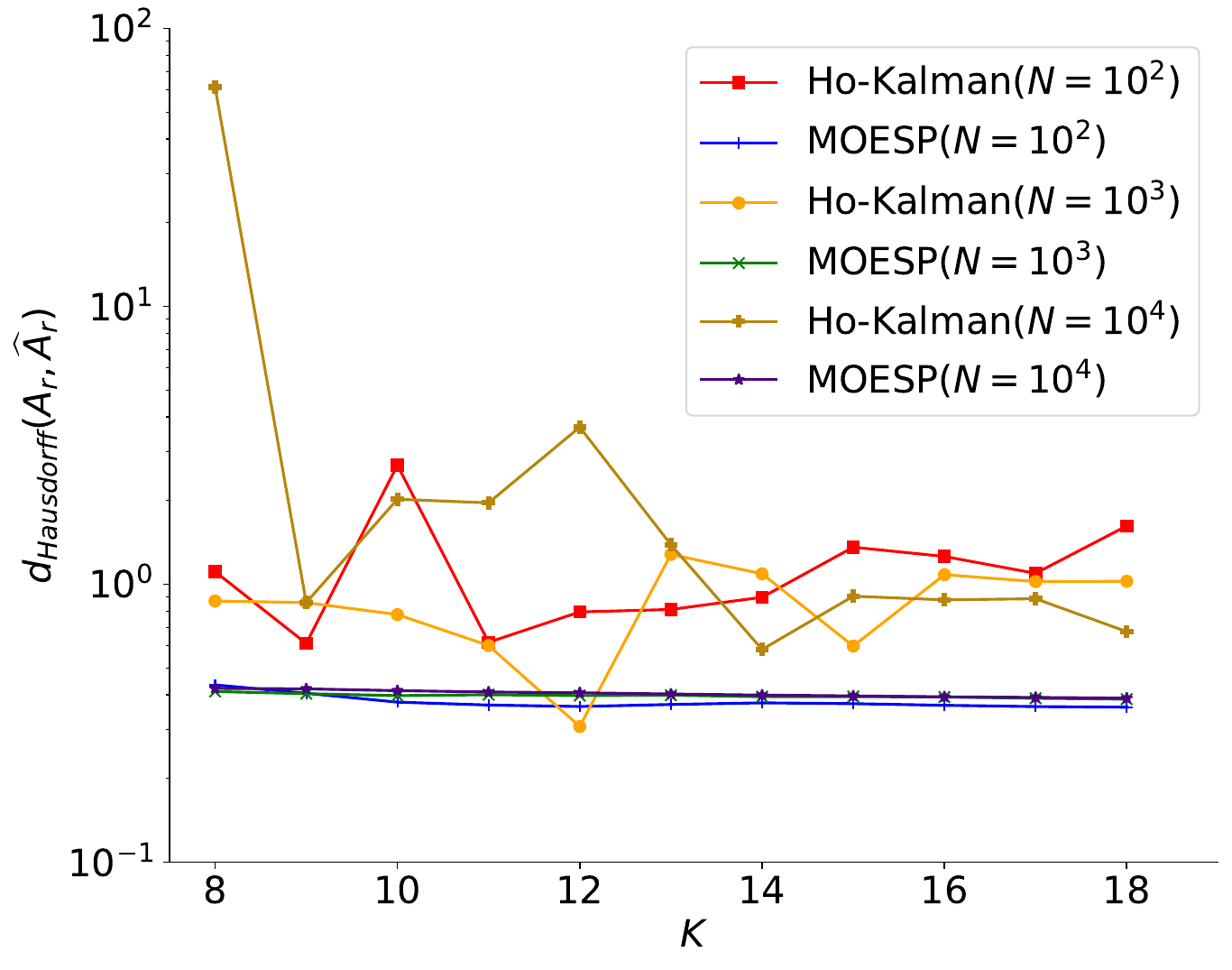} }
	\subcaptionbox{\label{pole_distribution}}[.22\linewidth]
	{\includegraphics[width = 1\linewidth]{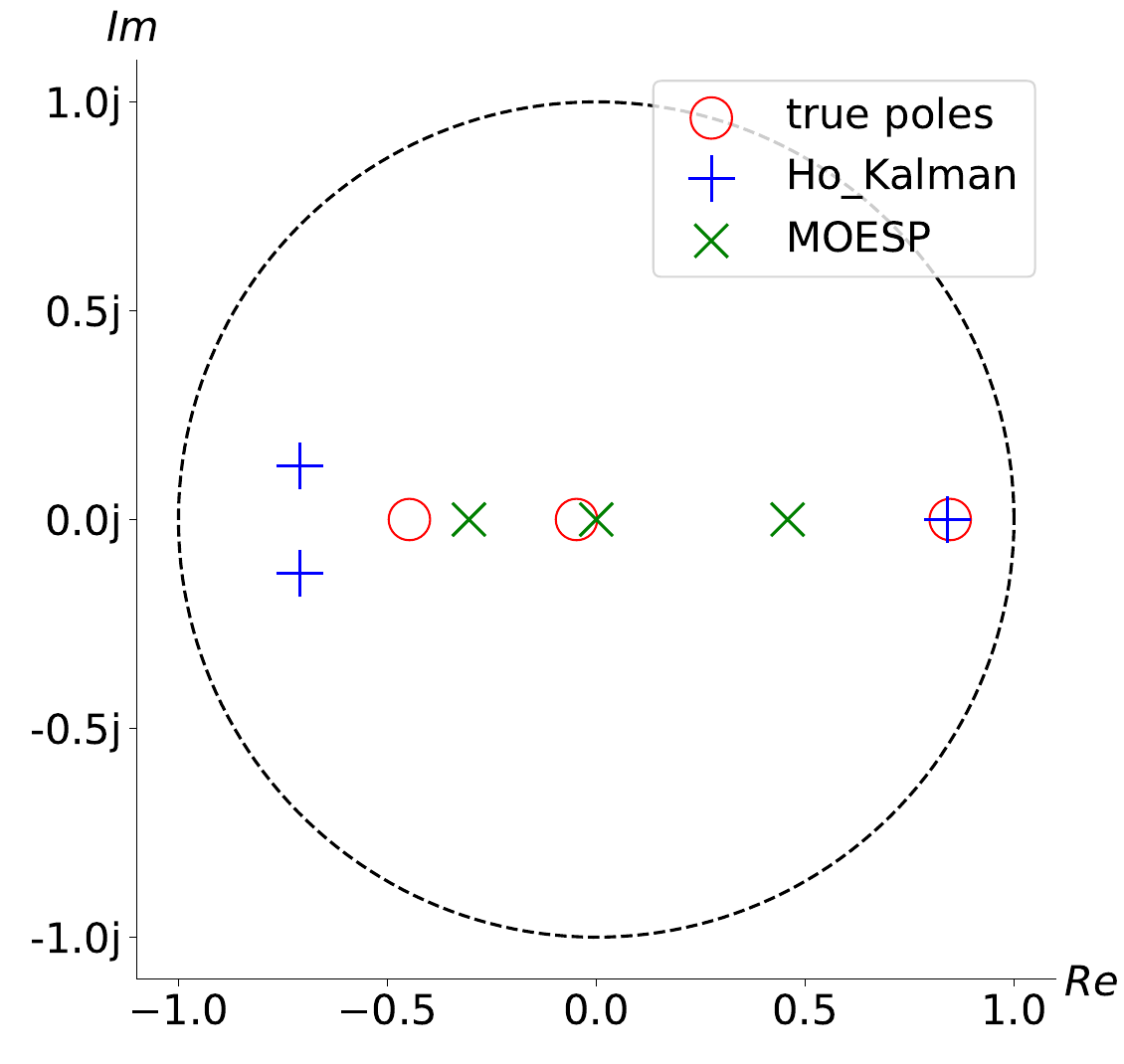} }
	\subcaptionbox{\label{pde_id_min_singular_value_H}}[.24\linewidth] {
		\includegraphics[width = 1\linewidth]{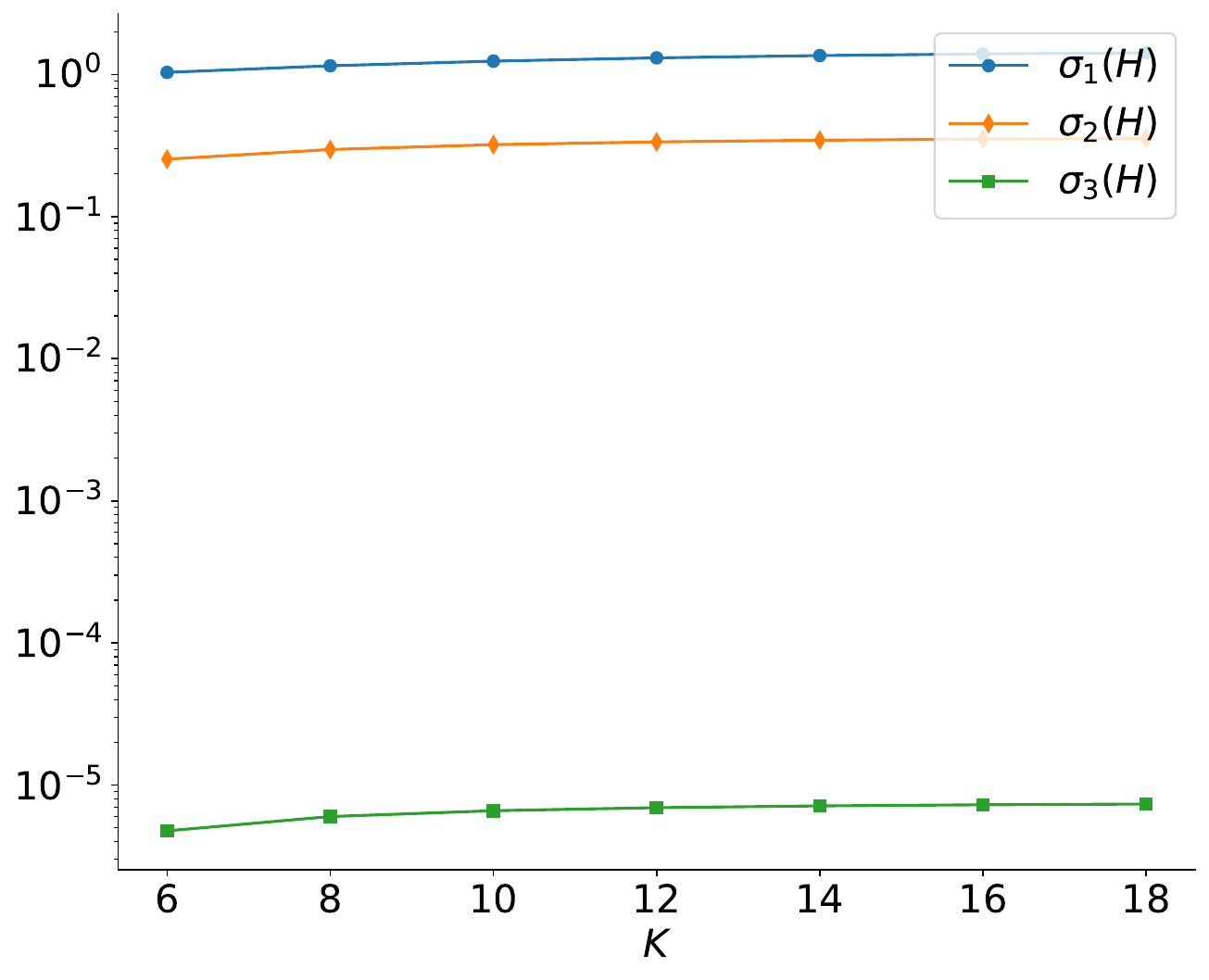}}
	\subcaptionbox{\label{pde_id_condition_number_O_Q}}[.24\linewidth] {
		\includegraphics[width = 1\linewidth]{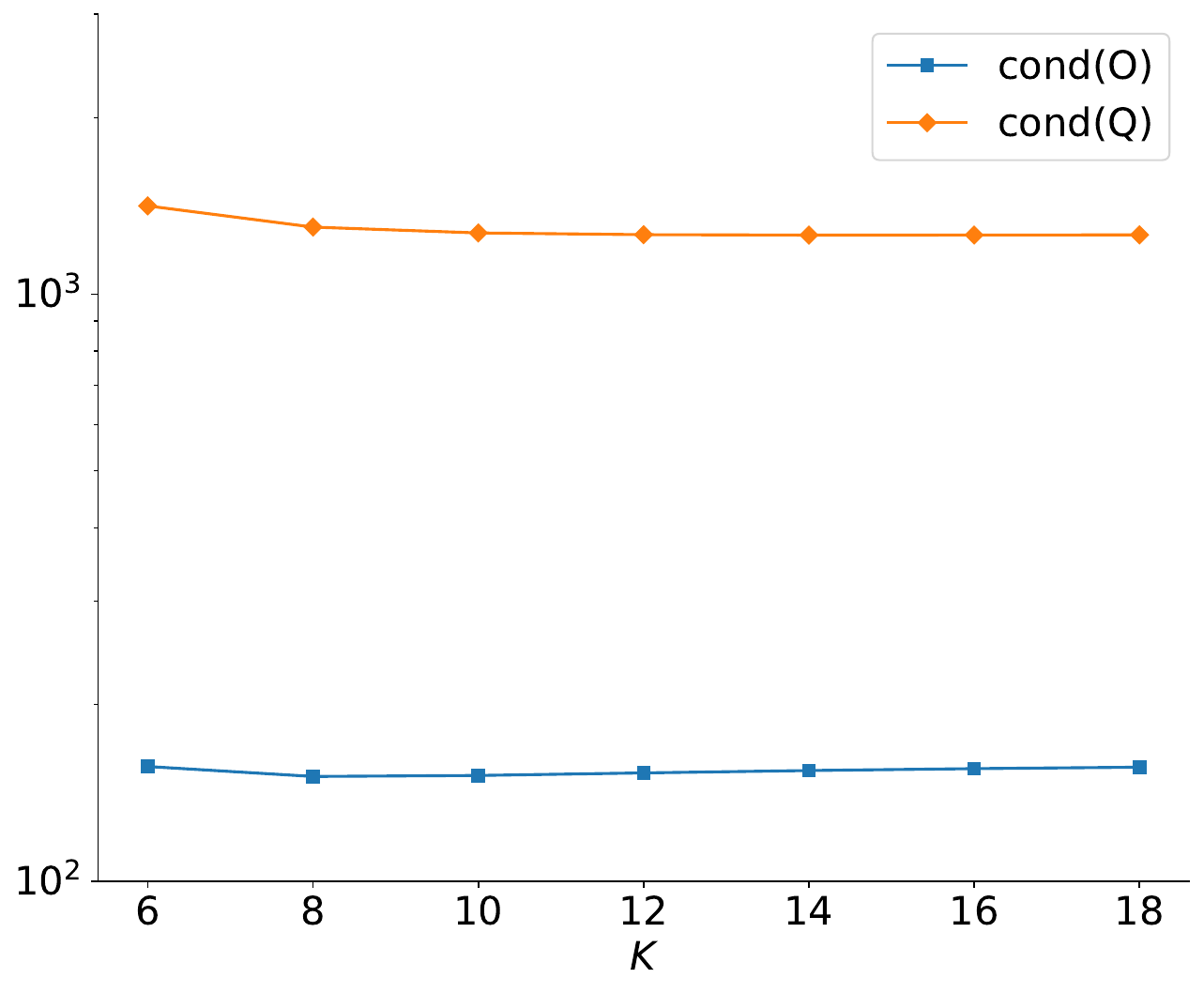}}
	\caption{\response{In sub-figure (a), 
			the red, orange and golden lines correspond to the identification accuracy of the Ho-Kalman algorithm when the number of trajectories $N$ is $10^2$, $10^3$, and $10^4$ respectively, and the blue, green and indigo lines correspond to the identification accuracy of the MOESP algorithm when the number of trajectories $N$ is $10^2$, $10^3$, and $10^4$ respectively. $K$ represents the length of trajectory used for identification.		Sub-figure (b) shows the distribution of the true poles (red `$\circ$') and the identified poles (blue `$+$' and green `$\times$' correspond to Ho-Kalman algorithm and MOESP algorithm respectively) in the complex plane when the number of trajectories $N$ is $10^4$ and the trajectory length $K$ is $18$, where the dashed line is the unit circle. The blue solid line in sub-figure (c) depicts the smallest singular value of the Hankel matrix  $\sigma_{\min}(H)$, and the blue and orange solid line in sub-figure (d) depict the condition numbers of (extended) observability and controllability matrices $O,Q$.
	}}
	\label{pde}
\end{figure*}

\response{
	We further assume that there are $4$ heat sources (sinks) serving as actuators at location $(l/4,l/4)$, $(l/4,3l/4)$, $(3l/4,l/4)$, and $(3l/4,3l/4)$. 
	If we combine the temperature of all grid points at time $k$ into a column vector $ T_k$, the evolution process of the discrete LTI system can be expressed as
	\begin{equation}\label{evolution process}
		T_{k+1} = AT_k + Bu_k.
	\end{equation}
	We assume there is a sensor deployed in the center $(l/2,l/2)$ of the grid, and it measures a linear combination of temperature at grid points around it. 
	Therefore,the measurement $y_k$ can be expressed as
	\begin{equation}\label{measurement process}
		y_k = CT_k + v_k,
	\end{equation}
	where $v_k \sim \mathcal{N}(0,\mathcal{R})$ denotes the i.i.d. measurement noise. For numerical experiments, we set the parameters to the following values:
	\begin{itemize}
		\item $\alpha = 0.2$; $l = 3$;
		\item $\mathcal{R} = 1$; $u_k \sim \mathcal{N}(0,1)$.
	\end{itemize}
	Since the system matrix $A$ contains double root and the system is neither controllable nor observable, a Kalman decomposition is performed to reduce the model into its minimal controllable and observable realization $(A_{r}, B_{r}, C_{r})$, which contains $3$ internal states. }

\response{
	Ho-Kalman algorithm and MOESP algorithm are used to identify the system matrix of the minimal realization from multiple independent trajectories collected. The result obtained is depicted in Figure~\ref{pde}. We use the Hausdorff distance to measure the distance between the true and the identified spectrum, which is defined as
	\begin{equation}
		\small d_{\rm Hausdorff}(A_{r},\widehat{A}_{r}) \triangleq \max \{
		\max_i \min_j |\lambda_i - \hat{\lambda}_j |, \max_j \min_i |\lambda_i - \hat{\lambda}_j| 
		\},
	\end{equation}
	where $\lambda_i$s and $\hat\lambda_i$s are the eigenvalues of $A_r$ and $\widehat A_r$ matrices respectively. It can be seen that the Hausdorff distances for both algorithms are all greater than $0.3$, and does not improve much even with more trajectories collected. Notice that for the minimal system, all the $3$ poles are real due to the symmetry of the $A$ matrix and the spectral radius is $0.85$. Hence, a trivial guess of $\widehat A_r= \mathbf{0}$ would provide a Hausdorff distance no more than $0.85$, which implies that even with $10^4$ trajectories collected, the Ho-Kalman algorithm and MOESP algorithm are not doing better than a trivial guess.
	On the other hand, it can be seen that $\sigma_{\min}(H)$ is minuscule and the condition numbers of (extended) observability and controllability matrices $O,Q$  are gargantuan, which shows that both identification algorithms are ill-conditioned. }

		\ssun{Next, we consider the identification of the classical two-mass spring-damper system \cite{liu2013fast}. The system consists of two point masses, $m_1 = 1$ kg and $m_2 = 1$ kg, interconnected in series via linear springs and viscous dampers. Each mass is also coupled to a fixed boundary via a spring-damper pair, as illustrated in Figure~\ref{fig:enter-label}.}
	
	\begin{figure}[htbp]
		\centering
		\begin{tikzpicture}[
			spring/.style={  
				thick,  
				decoration={  
					coil,  
					aspect=0.5,  
					segment length=1mm,  
					amplitude=1mm,  
					pre length=3mm,  
					post length=3mm  
				},
				decorate  
			},
			wall/.style={  
				thick,
				pattern=north east lines,  
				minimum height=1.5cm,  
				minimum width=1pt  
			},
			wheel/.style={  
				circle,
				fill=white,
				draw=black,
				thick,
				minimum size=4pt,  
				outer sep=0pt
			},
			ground/.style={  
				thick,
				pattern=north east lines,  
				minimum height=0.1cm,  
				anchor=north  
			},
			damper/.style={
				draw, 
				thick,
				line width=1pt,
				minimum width=0.3cm,
				minimum height=0.2cm,
				inner sep=0pt,
				outer sep=0pt
			}
			]
			
			\def\wallsep{7cm}     
			\def\masswidth{1.2cm} 
			\def\massheight{0.8cm} 
			\def\groundlevel{-1cm}  
			\def\wheeloffset{0.3cm} 
			\def\dampershift{0.3cm} 
			
			\node (ground) [ground, minimum width=\wallsep, anchor=north] at (\wallsep/2, \groundlevel-5) {};
			\draw[thick] (ground.north west) -- (ground.north east);
			
			\node (left wall) [wall, anchor=east] at (0, \groundlevel+16.5) {};  
			\draw[thick] (left wall.north east) -- (left wall.south east);  
			
			\node (right wall) [wall, anchor=west] at (\wallsep, \groundlevel+16.5) {};  
			\draw[thick] (right wall.north west) -- (right wall.south west);  
			
			\node (mass1) [draw, thick, minimum width=\masswidth, minimum height=\massheight, anchor=south] at (2, \groundlevel) {$m_1$};
			\fill (1.4,\massheight/2+\groundlevel+6) circle (1pt);  
			\fill (2.6,\massheight/2+\groundlevel+6) circle (1pt);  
			\fill (1.4,\massheight/2+\groundlevel-6.5) circle (1pt); 
			\fill (2.6,\massheight/2+\groundlevel-6.5) circle (1pt); 
			
			\fill (0,\massheight/2+\groundlevel+6) circle (1pt); 
			\fill (0,\massheight/2+\groundlevel-6.5) circle (1pt); 
			\fill (7,\massheight/2+\groundlevel+6) circle (1pt); 
			
			\node [wheel] at ([xshift=-\wheeloffset]mass1.south) {};  
			\node [wheel] at ([xshift=\wheeloffset]mass1.south) {};   
			
			\node (mass2) [draw, thick, minimum width=\masswidth, minimum height=\massheight, anchor=south] at (5, \groundlevel) {$m_2$};
			\fill (4.4,\massheight/2+\groundlevel+6) circle (1pt);  
			\fill (5.6,\massheight/2+\groundlevel+6) circle (1pt);  
			\fill (4.4,\massheight/2+\groundlevel-6.5) circle (1pt); 
			
			\node [wheel] at ([xshift=-\wheeloffset]mass2.south) {};  
			\node [wheel] at ([xshift=\wheeloffset]mass2.south) {};   
			
			\draw[spring] (0,\massheight/2+\groundlevel+6) -- (1.4,\massheight/2+\groundlevel+6)
			node [midway, above=0.5mm] {\small $k_1$};
			
			\draw[spring] (2.6,\massheight/2+\groundlevel+6) -- (4.4,\massheight/2+\groundlevel+6)
			node [midway, above=0.5mm] {\small $k_2$};
			
			\draw[spring] (5.6,\massheight/2+\groundlevel+6) -- (7,\massheight/2+\groundlevel+6)
			node [midway, above=0.5mm] {\small $k_3$};
			
			\coordinate (dampStart1) at ([yshift=-\dampershift]0,\massheight/2+\groundlevel+2);
			\coordinate (dampEnd1) at ([yshift=-\dampershift]1.4,\massheight/2+\groundlevel+2);
			
			\node [damper, anchor=center] (d1) at ($(dampStart1)!0.5!(dampEnd1)$) {};
			
			\draw[thick] (dampStart1) -- (d1.west);
			\draw[thick] (d1.east) -- (dampEnd1);
			
			\draw[thick] (d1.north east) -- ++(0.1cm,0);
			\draw[thick] (d1.south east) -- ++(0.1cm,0);

			\node[below] at (1.1,-0.8) {\small $c_1$};
			
			\coordinate (dampStart2) at ([yshift=-\dampershift]2.6,\massheight/2+\groundlevel+2);
			\coordinate (dampEnd2) at ([yshift=-\dampershift]4.4,\massheight/2+\groundlevel+2);
			
			\node [damper, anchor=center] (d2) at ($(dampStart2)!0.5!(dampEnd2)$) {};
			
			\draw[thick] (dampStart2) -- (d2.west);
			\draw[thick] (d2.east) -- (dampEnd2);
			
			\draw[thick] (d2.north east) -- ++(0.1cm,0);
			
			\draw[thick] (d2.south east) -- ++(0.1cm,0);

			\node[below] at (1.1+2.8,-0.8) {\small $c_2$};

			\draw[thick, ->,black] (2,-0.2+0.2) -- ++(0.8,0) 
			node[above] {$q_1$};
			\draw[thick, ->, black] (2-0.8,-0.2+0.4) -- ++(0.8,0) 
			node[pos=0.2, above] {$f_1$};
			\draw[thick, ->, black] (5,-0.2+0.2) -- ++(0.8,0) 
			node[above] {$q_2$};
			\draw[thick, ->, black] (5-0.8,-0.2+0.4) -- ++(0.8,0) 
			node[pos=0.2, above] {$f_2$};
			
			\draw[thick] (2,-0.2) -- (2,-0.2+0.6);
			
			\draw[thick] (5,-0.2) -- (5,-0.2+0.6
			);

		\end{tikzpicture}
		\caption{\ssun{A two-mass spring-damper system.}}
		\label{fig:enter-label}
	\end{figure}
	
	\ssun{Let $q_1(t)$ and $q_2(t)$ denote the horizontal displacements of the masses $m_1$ and $m_2$, respectively. The control forces $f_1(t)$ and $f_2(t)$ are applied directly to the respective masses, while the system outputs are defined as their displacements. Applying Newton's second law yields the following equations of motion:
		\begin{align}\notag
			m_1 \ddot{q}_1 &= -k_1 q_1 + k_2(q_2 - q_1) - c_1 \dot{q}_1 + c_2(\dot{q}_2 - \dot{q}_1) + f_1(t), \\\notag
			m_2 \ddot{q}_2 &= -k_2(q_2 - q_1) - k_3 q_2 - c_2(\dot{q}_2 - \dot{q}_1) + f_2(t).
		\end{align}
		Defining the state vector, control input vector and output vector as $x(t) = \begin{bmatrix}
			q_1(t) & \dot{q}_1(t) & q_2(t) & \dot{q}_2(t)
		\end{bmatrix}^\top$, $u(t) = \begin{bmatrix}
			f_1(t) & f_2(t)
		\end{bmatrix}^\top$, $y(t) = \begin{bmatrix}
			q_1(t)  & q_2(t) 
		\end{bmatrix}^\top$,
		the system dynamics can be expressed in the continuous-time state-space form:
		\begin{align}\notag
			\dot{x}(t) &= A_cx(t) + B_cu(t), \\\notag
			y(t) &= C_cx(t),
		\end{align}
		where 
		$
		A_c = \begin{bmatrix}
			0 & 1 & 0 & 0 \\
			-\frac{k_1 + k_2}{m_1} & -\frac{c_1 + c_2}{m_1} & \frac{k_2}{m_1} & \frac{c_2}{m_1} \\
			0 & 0 & 0 & 1 \\
			\frac{k_2}{m_2} & \frac{c_2}{m_2} & -\frac{k_2 + k_3}{m_2} & -\frac{c_2}{m_2}
		\end{bmatrix}, \quad
		B_c = \begin{bmatrix}
			0 & \frac{1}{m_1} & 0 & 0\\
			0 & 0 & 0 & \frac{1}{m_2}
		\end{bmatrix}^\top, \quad
		C_c = \begin{bmatrix}
			1 & 0 & 0 & 0 \\
			0 & 0 & 1 & 0
		\end{bmatrix}
		$.}
	
	\ssun{To facilitate identification and digital implementation, the continuous-time system is converted to a discrete-time system using zero-order hold (ZOH) with a sampling period of $T_s = 0.1$ seconds. The resulting discrete-time state-space model is:
		\begin{equation}\label{discrete-time system}
			\begin{aligned}
				x_{k+1} &= A_d x_k + B_d u_k, \\
				y_k &= C_d x_k + v_k,
			\end{aligned}
		\end{equation}
		where $A_d = e^{A_c T_s}$, $B_d = \left( \int_0^{T_s} e^{A_c \tau} d\tau \right) B_c$, $C_d = C_c$, and  $v_k \sim \mathcal{N}(\mathbf{0},10^{-2}I_2)$ is the i.i.d. measurement noise. For numerical experiments, we set the parameters to the following values: 
		\begin{itemize}
			\item $k_1 = 0.8$, $k_2 = 1.2$, $k_3 = 1.5$;
			\item  $c_1 = c_2 = 7$; $u_k \sim \mathcal{N}(\mathbf{0},I_2)$.
	\end{itemize}}
	\ssun{Ho-Kalman algorithm and MOESP algorithm are used to identify the system matrix from multiple independent trajectories collected. The result obtained is depicted in Figure~\ref{two_mass}. It can be seen that the Hausdorff distances for both algorithms are all greater than $0.3$, and does not improve much even with more trajectories collected.  Notice that the spectrum of matrix $A_d$ is $\{0.16,0.83,0.92,0.99\}$. Hence, a trivial guess of $\widehat A_d= \mathbf{0}$ would provide a Hausdorff distance no more than $0.99$, which implies that even with $10^4$ trajectories collected, the Ho-Kalman algorithm and MOESP algorithm are not doing better than a trivial guess. On the other hand, it can be seen that $\sigma_{\min}(H)$ is minuscule and the condition numbers of (extended) observability and controllability matrices $O,Q$  are gargantuan, which shows that both identification algorithms are ill-conditioned.}
	
	\begin{figure*}[htbp]
		\centering
		\subcaptionbox{\label{two_mass_id}}[.26\linewidth]
		{\includegraphics[width = 1\linewidth]{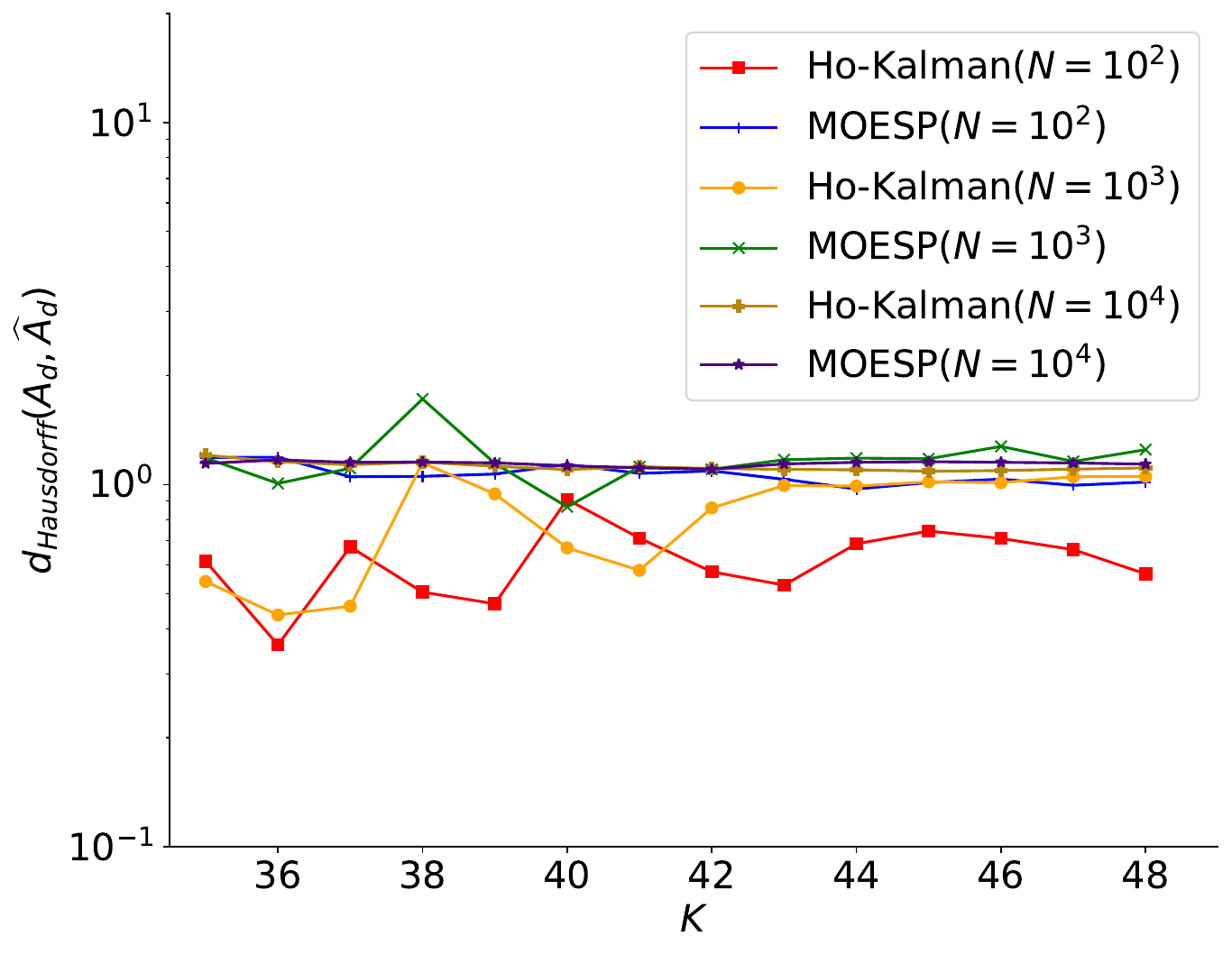} }
		\subcaptionbox{\label{two_mass_pole_distribution}}[.22\linewidth]
		{\includegraphics[width = 1\linewidth]{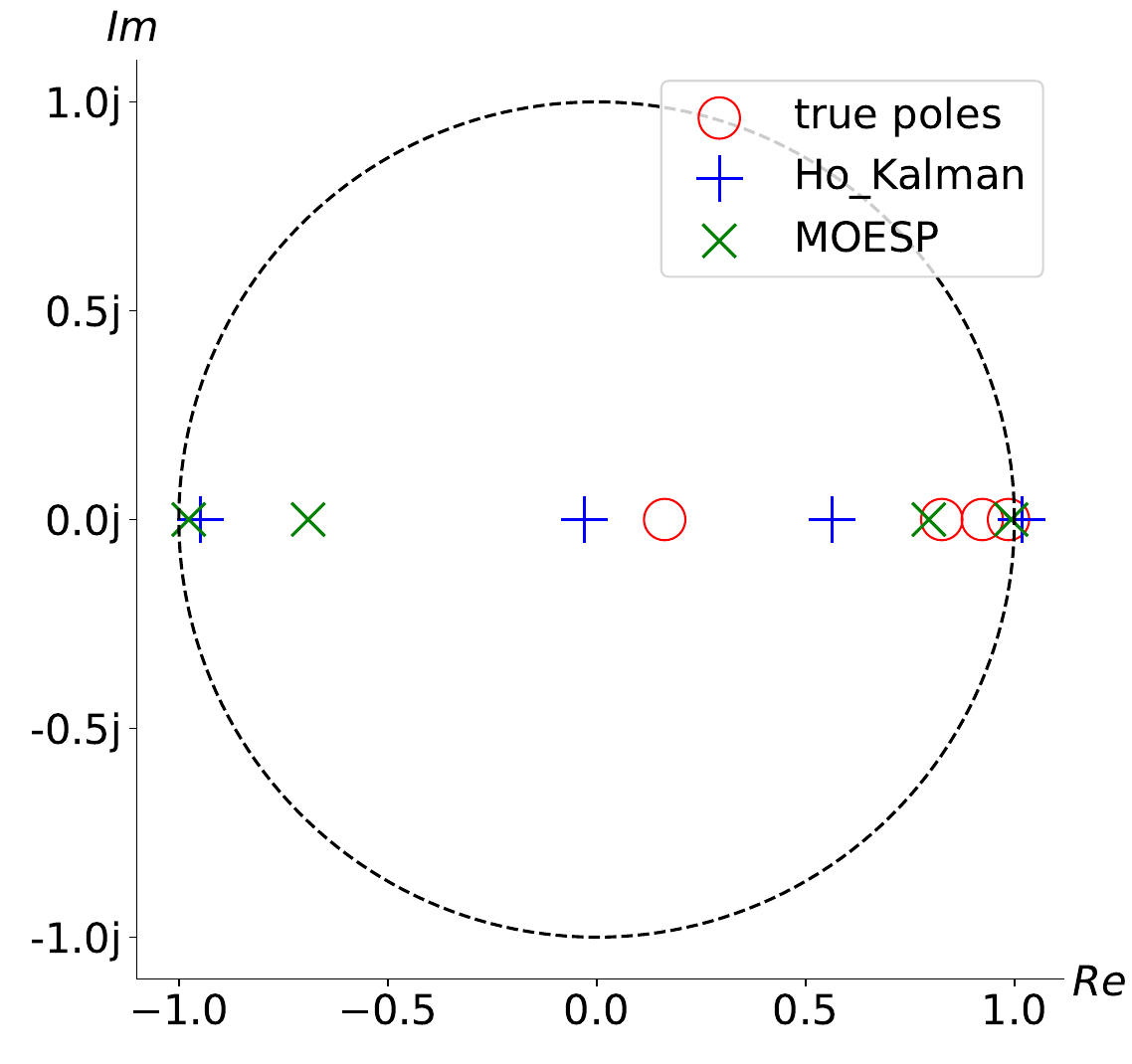} }
		\subcaptionbox{\label{two_mass_id_min_singular_value_H}}[.24\linewidth] {
			\includegraphics[width = 1\linewidth]{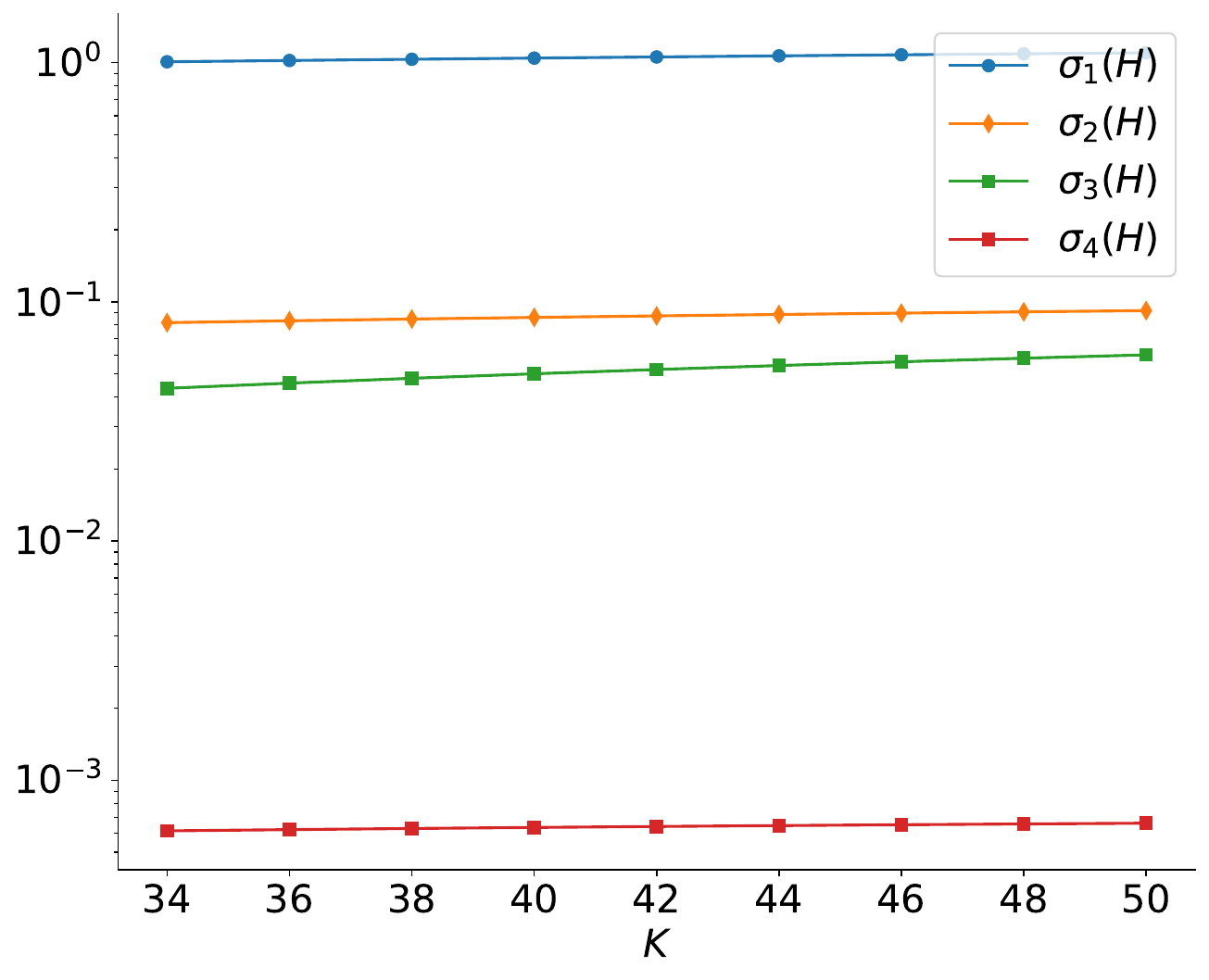}}
		\subcaptionbox{\label{two_mass_id_condition_number_O_Q}}[.24\linewidth] {
			\includegraphics[width = 1\linewidth]{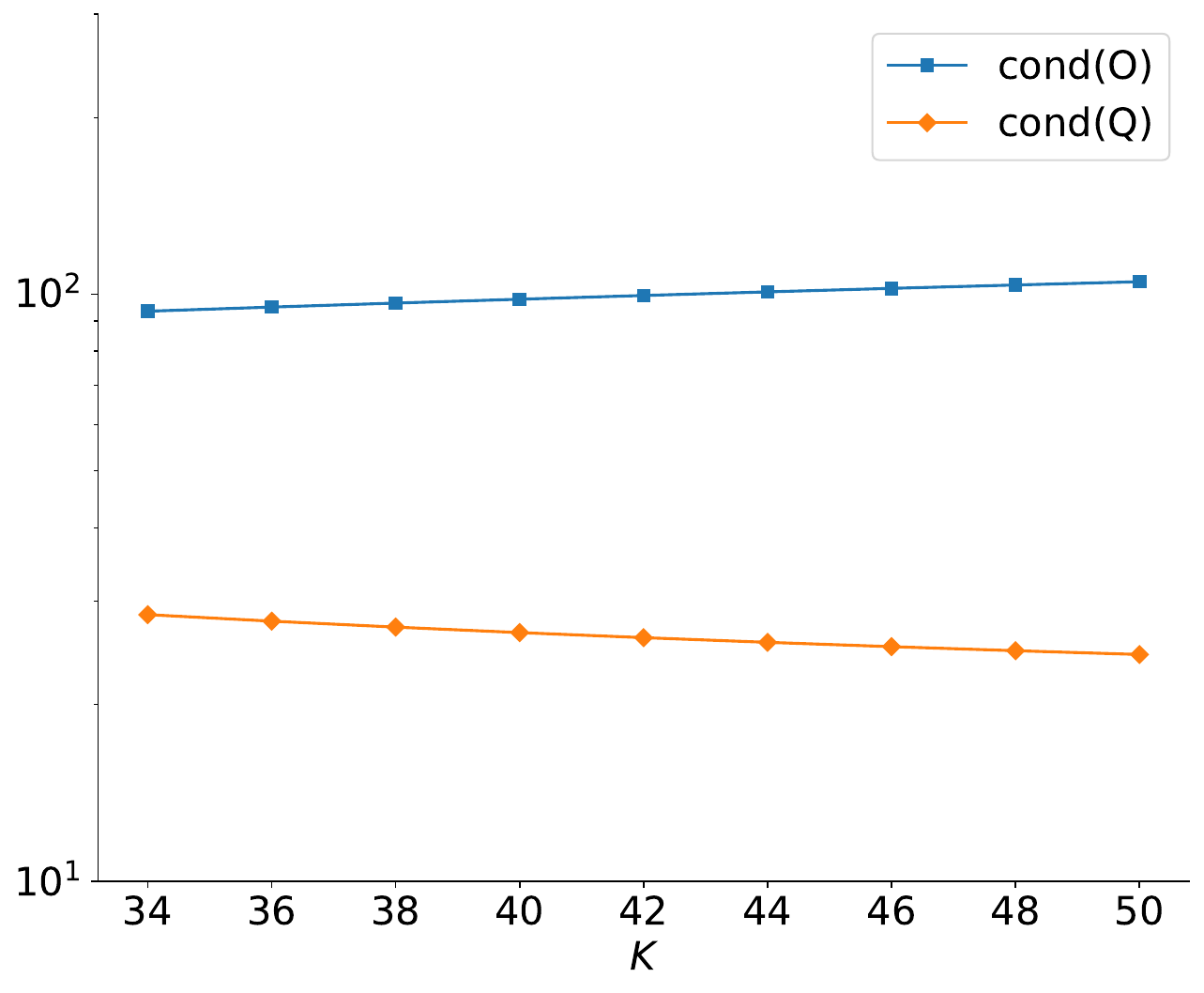}}
		\caption{\response{\ssun{In sub-figure (a), 
					the red, orange and golden lines correspond to the identification accuracy of the Ho-Kalman algorithm when the number of trajectories $N$ is $10^2$, $10^3$, and $10^4$ respectively, and the blue, green and indigo lines correspond to the identification accuracy of the MOESP algorithm when the number of trajectories $N$ is $10^2$, $10^3$, and $10^4$ respectively. $K$ represents the length of trajectory used for identification.		Sub-figure (b) shows the distribution of the true poles (red `$\circ$') and the identified poles (blue `$+$' and green `$\times$' correspond to Ho-Kalman algorithm and MOESP algorithm respectively) in the complex plane when the number of trajectories $N$ is $10^4$ and the trajectory length $K$ is $48$, where the dashed line is the unit circle. The blue solid line in sub-figure (c) depicts the smallest singular value of the Hankel matrix  $\sigma_{\min}(H)$, and the blue and orange solid line in sub-figure (d) depict the condition numbers of (extended) observability and controllability matrices $O,Q$.}
		}}
		\label{two_mass}
	\end{figure*}

	In the next numerical example, we consider the estimation of $c_i, \lambda_i$ from the noisy signal $y_k = \sum_{i=1}^n c_i \lambda_i^k+ v_k$, where the noise $v_k \sim \mathcal N(0, 1)$. This problem can be solved by cast it into the identification of LTI system \eqref{eq:noisysignalest}. As a result, we first consider solving the estimation problem via Ho-Kalman algorithm or MOESP algorithm, and show the $n$-th largest singular value of the Hankel matrix $ H$, the condition number of the observability matrix $ O $ used in both algorithms, and the corresponding bounds derived in Theorem~\ref{Ho-Kalman is ill-conditioned} in Figure~\ref{fig-L} and Figure~\ref{fig-O}. $ K_1 $, $ K_2 $ are set to $ n $. For each $n$, we repeat $ 10^4 $ trials, with $\lambda_i, c_i$ both uniformly sampled from the line segment $ [-1,1] $.

\begin{figure*}[htbp]
	\centering
	\begin{minipage}[b]{0.48\textwidth}
		\includegraphics[width=\textwidth]{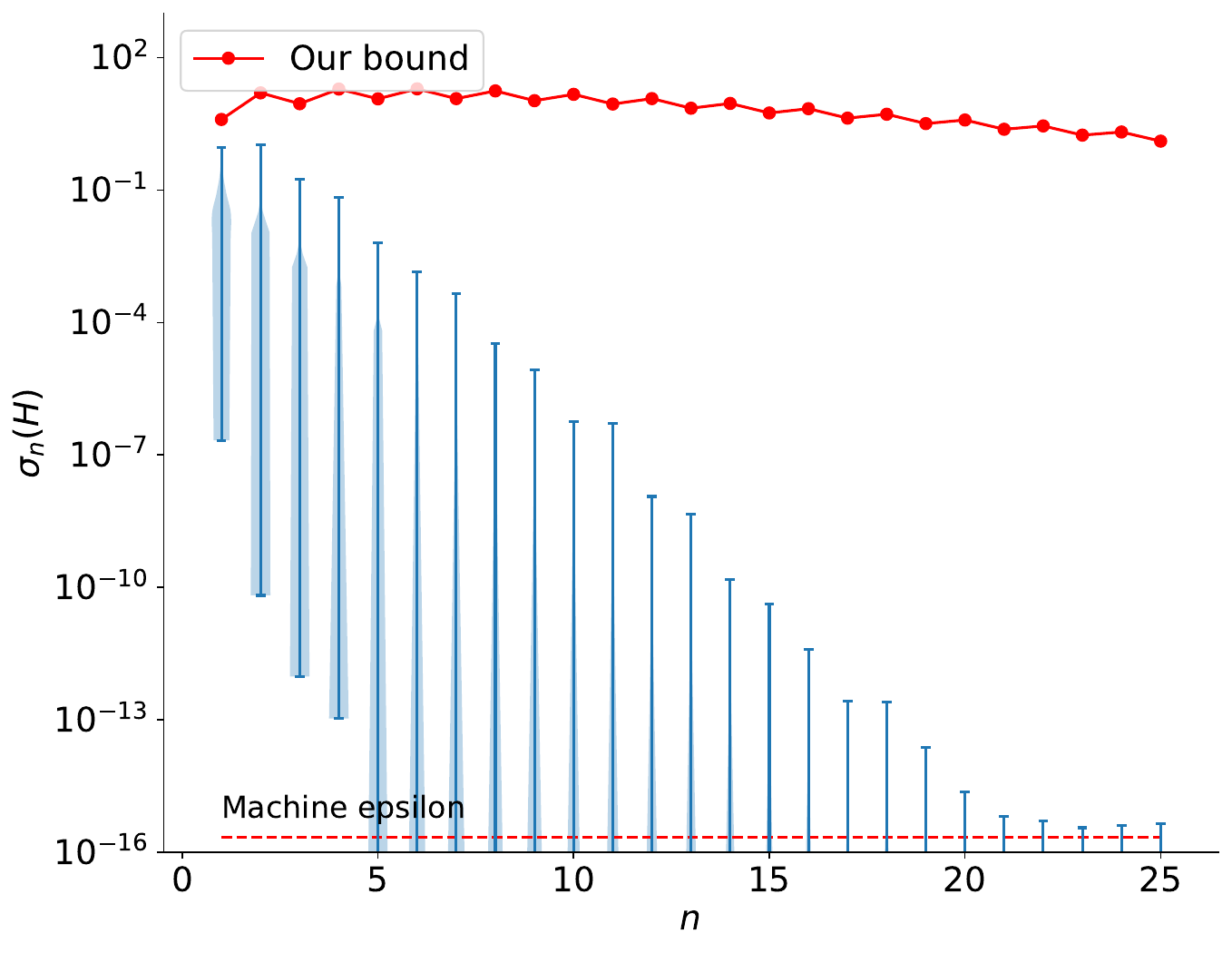}
		\caption{Numerical results of the key matrix $ H $ in Ho-Kalman algorithm.
			The blue violinplot represents the sampling distribution of the $n$-th largest singular value of the Hankel matrix $ H$ in different dimension $ n $.
			The red dashed and solid lines represent machine epsilon and the bound we derive in \eqref{L}, respectively.}
		\label{fig-L}
	\end{minipage}
	\hfill
	\begin{minipage}[b]{0.48\textwidth}
		\includegraphics[width=\textwidth]{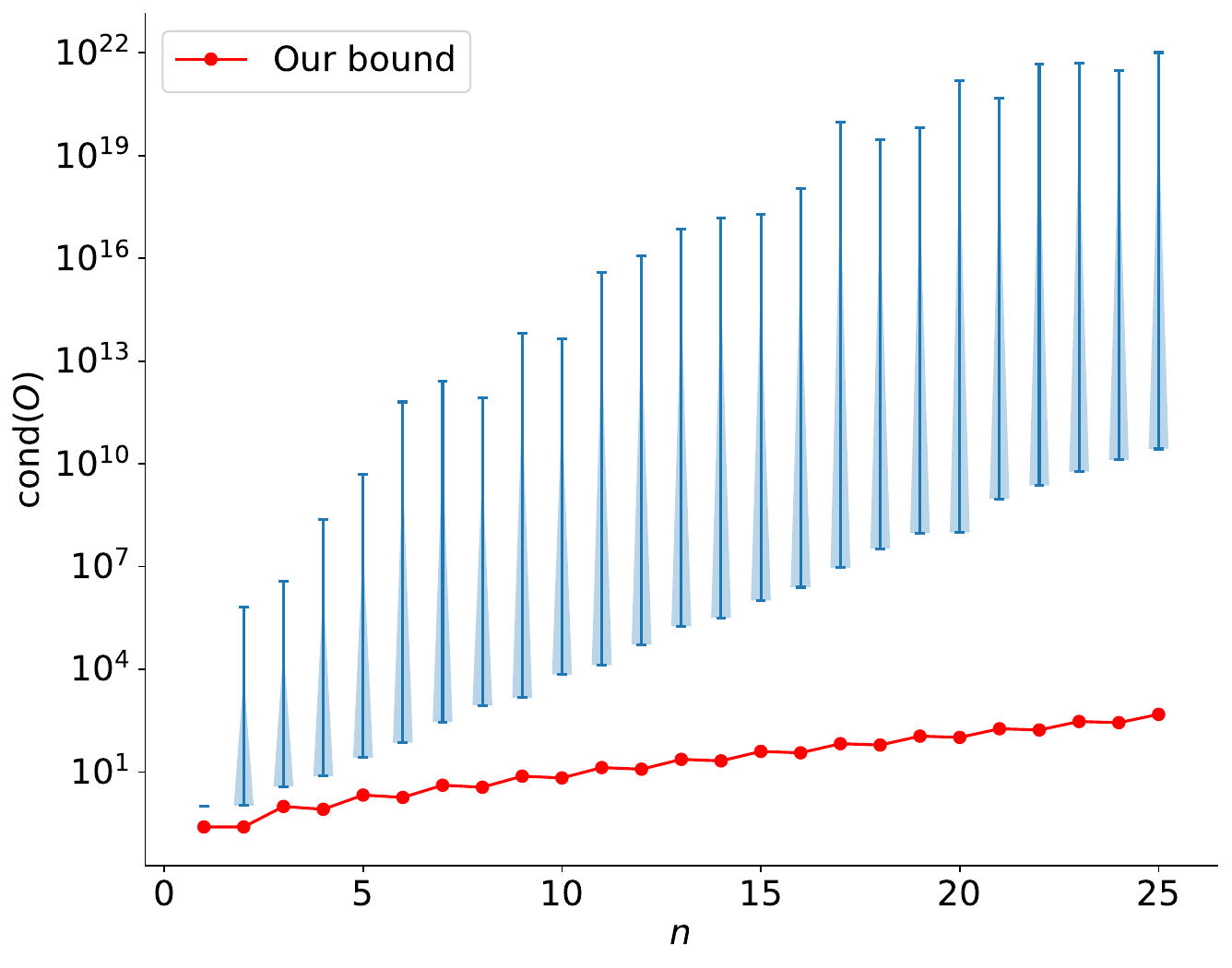}
		\caption{Numerical results of the key matrix $ O $ in Ho-Kalman/MOESP algorithm.
			The blue violinplot represents the sampling distribution of the condition number of the observability matrix $O $  in different dimension $ n $.
			The red solid line represents  the bound we derive in \eqref{O and Q}.}
		\label{fig-O}
	\end{minipage}
\end{figure*}

\begin{figure*}[bthp]
	\centering
	\begin{subfigure}[t]{0.33\linewidth}
		\centering
		\includegraphics[width=\linewidth]{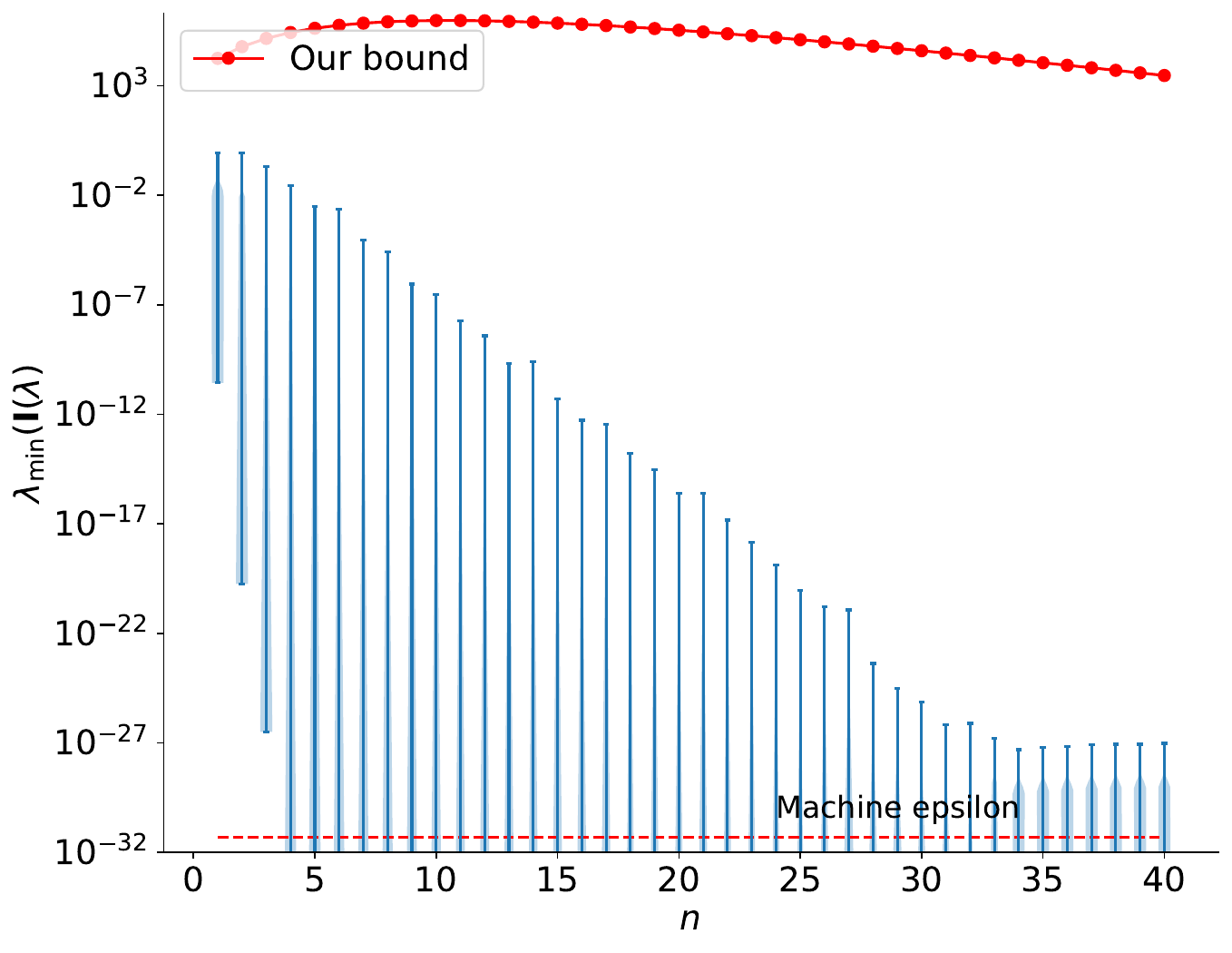}
		\caption{\ssun{The trajectory length $K = n+1$}}
		\label{fig/eig_min_fisher}
	\end{subfigure}%
	\hfill
	\begin{subfigure}[t]{0.33\linewidth}
		\centering
		\includegraphics[width=\linewidth]{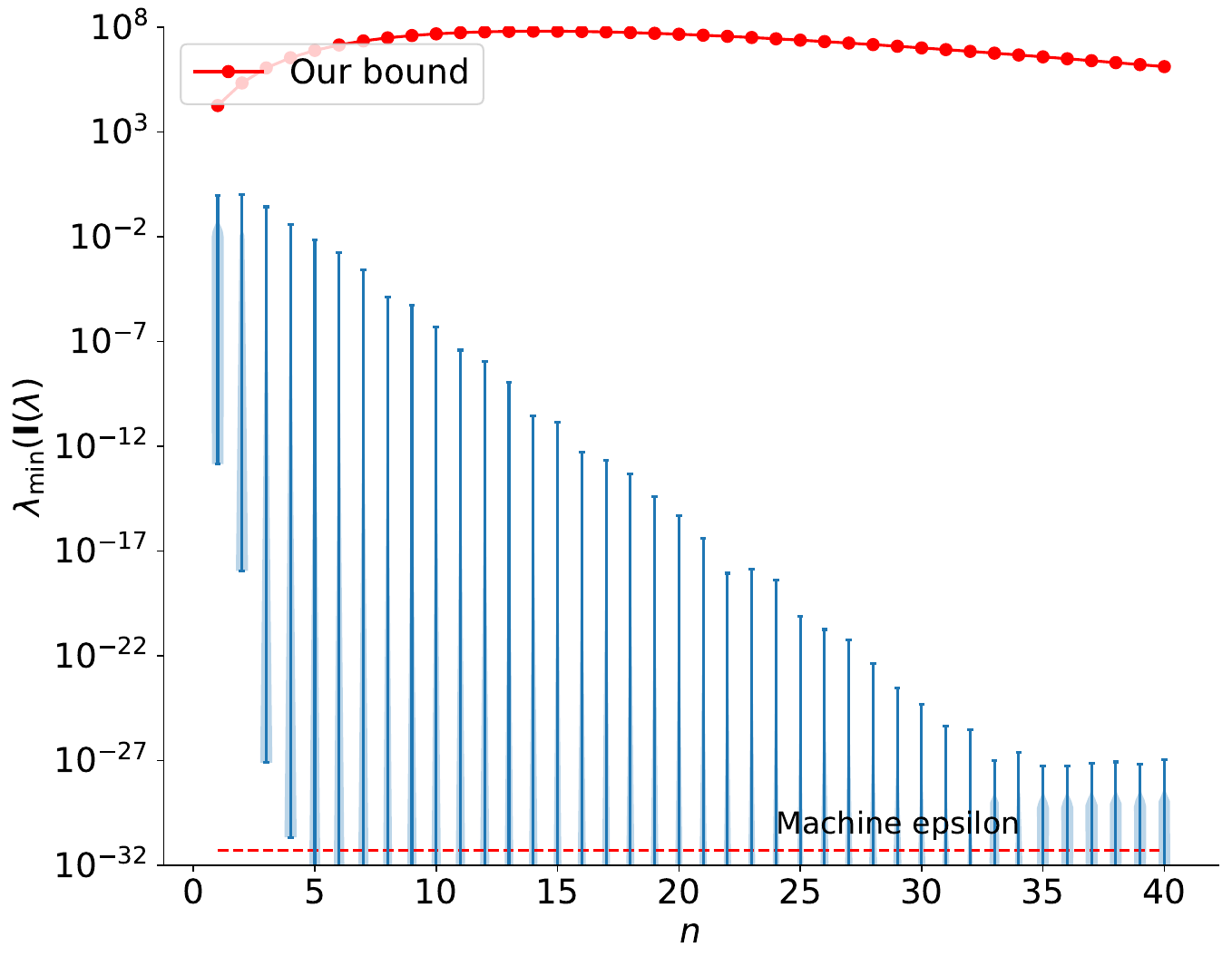}
		\caption{\ssun{The trajectory length $K = 2n$}}
		\label{fig/eig_min_fisher_2n}
	\end{subfigure}%
	\hfill
	\begin{subfigure}[t]{0.33\linewidth}
		\centering
		\includegraphics[width=\linewidth]{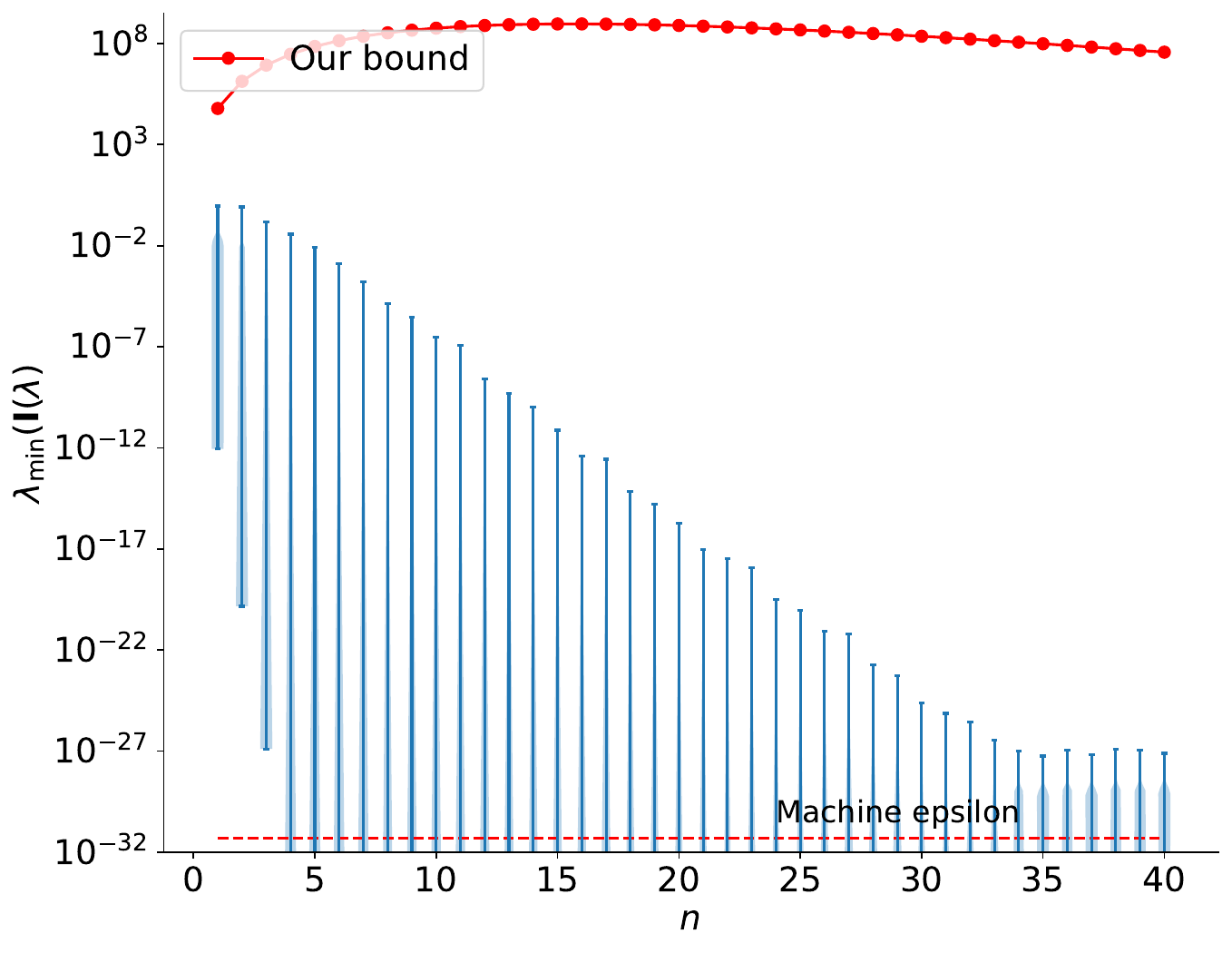}
		\caption{\ssun{The trajectory length $K = 3n$}}
		\label{fig/eig_min_fisher_3n}
	\end{subfigure}
	\caption{\ssun{Numerical results of Fisher Information matrix $ \mathbf{I}(\bm{\lambda}) $ under different trajectory lengths: (a) $K = n+1$; (b) $K = 2n$; (c) $K = 3n$.
		The blue violinplot represents the sampling distribution of the smallest eigenvalue of Fisher Information matrix $ \mathbf{I}(\bm{\lambda}) $ in different dimension $ n $.
			The red dashed and solid lines represent machine epsilon and the bound we derive in \eqref{worst case upper bound}, respectively.}}
	\label{fig_fisher}
\end{figure*}

Finally, to show that the identification problem is ill-conditioned, i.e. all unbiased estimators of $\lambda_i$s are ill-conditioned, we depict\footnote{The machine epsilon in Figure \ref{fig_fisher} is different from that in Figure \ref{fig-L}, because we do not directly calculate $ \lambda_{n}(\mathbf{I}(\bm{\lambda})) $, but according to \eqref{I = SV}, first calculate $ \sigma_{n}(SV) $, and then square it.} the smallest eigenvalue of the Fisher Information matrix $ \mathbf{I}(\bm{\lambda}) $ and our upper bound derived in Theorem~\ref{worst case}  in Figure~\ref{fig_fisher}. Without loss of generality, we set the number $ N $ of trajectories\footnote{The Fisher information matrix is linear with respect to the number of independently collected trajectories.} to 1.
\ssun{The trajectory length $K$ is chosen as $n+1$, $2n$, or $3n$, corresponding to the cases illustrated in Figs. ~\ref{fig/eig_min_fisher}, ~\ref{fig/eig_min_fisher_2n}, and~\ref{fig/eig_min_fisher_3n}, respectively.} 
For each $n$, we repeat $ 10^4 $ trials.

It can be seen that our bounds are valid, However, the sampled system is much more ill-conditioned than our results suggested. Since our ``uniform'' bound works for any system, there may exist some parameters $c_i, \lambda_i$ which make the identification problem ``less'' ill-conditioned, but is not sampled due to the limited sample size. On the other hand, we think that there may be ways to strengthen our bounds, which we plan to investigate in the future.

\section{Conclusion}\label{sec:conclusion}

This paper considers the finite \syl{sample} identification performance of an $ n $ dimensional discrete-time MIMO LTI system, with $p$ inputs and $m$ outputs. We first prove that the widely-used Ho-Kalman algorithm and \revise{MOESP} algorithm are ill-conditioned for MIMO systems when $ n/m$ or $n/p $ is large. Moreover, a fundamental limit on MIMO systems identification is derived by analyzing the Fisher Information Matrix used in Cram\'er–Rao bound. Based on this analysis, we reveal that the sample complexity on unknown poles of stable (or marginally stable) MIMO systems using \textit{any} unbiased estimation algorithm to a certain level of accuracy explodes superpolynomially with respect to $ n/(pm)$. Future works include tightening our bound while designing system identification algorithms that can approach the fundamental limit.


\bibliographystyle{unsrt}
\bibliography{ref}{}

\section{Appendix}
\subsection{Preliminaries}
\begin{lemma}\label{krylov matrix}
	Given a $ n \times mp $ ($ p \leq  n $) matrix $ X_{n,mp} $ satisfying 
	\begin{equation}
		X_{n,mp} = \begin{bmatrix}
			W_p &
			DW_p &
			\cdots &
			D^{m-1}W_p
		\end{bmatrix},
	\end{equation}	
	where $ W_p \in \mathbb{R}^{n \times p} $ and $ D $ is a \ssun{normal} matrix with real eigenvalues, then the smallest singular value of $ X_{n,mp} $ satisfies
	\begin{equation}\label{krylov matrix_result}
		\sigma_{\min}(X_{n,mp}) \leq 4 \rho_{}^{-\frac{\left\lfloor \frac{\min\{n,m(p-[p]_*)\}-1}{2p} \right\rfloor}{\log (2mp)}} \|X_{n,mp}\|, 
	\end{equation}
	where $ \rho \triangleq e^{\frac{\pi^2}{4}} $, and $ [p]_* = 0 $ if $ p $ is even or $ p = 1 $ and is $ 1 $ if $ p $ is an odd number greater than $ 1 $.
\end{lemma}

\begin{proof}
	If $ p $ is equal to $ 1 $, the proof can be directly completed based on Corollary 5.3 in~\cite{beckermann2017singular}, thus in the following we  consider the case where $ p $ is greater than $ 1 $.
	It is not difficult to verify that matrix $ X_{n,mp} $ satisfies the following Sylvester matrix equation:
	\begin{equation}\label{sylvester matrix equation}
		DX_{n,mp} - X_{n,mp}P = \begin{bmatrix}
			\mathbf{0}_{n\times (m-1)p} & D^mW_p - W_pD_1
		\end{bmatrix},
	\end{equation} 
	where $ P = \begin{bmatrix}
		& D_1 \\ I_{(m-1)p} &
	\end{bmatrix} $, and $ D_1 = \begin{bmatrix}
		1 & \\ & -I_{p-1}
	\end{bmatrix} $.
	
	The following analysis is somewhat similar to the content of Section 5.1 in~\cite{beckermann2017singular}. For completeness, here we give a complete proof.
	For the following analysis, it can be assumed that $p$ is even. This is without loss of generality, since we can use Cauchy interlace theorem~\cite{hwang2004cauchy}. To see this, if $ p $ is an odd number greater than $ 1 $, let matrix $ W_{p-1} $ be the $ n \times (p-1) $ matrix obtained from $ W_{p} $ by removing its last column and $ X_{n,m(p-1)} $ be 
	\[
	\begin{bmatrix}
		W_{p-1}  & DW_{p-1} & \cdots & D^{m-1}W_{p-1}
	\end{bmatrix}.
	\]
	According to Cauchy interlace theorem~\cite{hwang2004cauchy}, it can be obtained that
	\[
	\sigma_{\min}(X_{n,mp}) \leq \sigma_{\min}(X_{n,m(p-1)}) \leq \|X_{n,m(p-1)}\| \leq \|X_{n,mp}\|,
	\]
	this implies that 
	\begin{equation}\label{even and odd}
		\frac{\sigma_{\min}(X_{n,mp})}{\|X_{n,mp}\|} \leq \frac{\sigma_{\min}(X_{n,m(p-1)}) }{\|X_{n,m(p-1)}\|}.
	\end{equation}
	
	In the following, we  assume that $p$ is even.
	Note that matrices $ D $ and $ P$ in Sylvester matrix equation \eqref{sylvester matrix equation} are both normal matrices. The eigenvalues of $ D $ are real numbers, and the eigenvalues of $ P $ are the $ mp $ (shifted) roots of unity, i.e, 
	\[
	\lambda(P) = \left\{
	z \in \mathbb{C} \mid z^{mp} = -1
	\right\}.
	\] 
	Since $ p $ is even, the spectrum $ \lambda(P) $ of $ P $ does not intersect the real axis. On the other hand, the rank of $  D^mW_p - W_pD_1 $ does not exceed $ p $. Applying Theorem 2.1, Corollary 3.2 and Lemma 5.1 in \cite{beckermann2017singular} yields that 
	\begin{equation}\label{even}
		\sigma_{\min}(X_{n,mp}) \leq 4 \rho^{-\frac{\left\lfloor \frac{\min\{n,mp\}-1}{2p} \right\rfloor}{\log (2mp)}} \|X_{n,mp}\|.
	\end{equation}
	If $ p $ is odd, then the proof can be completed according to \eqref{even and odd} and \eqref{even}.
\end{proof}

\begin{remark}
	\ssun{The result of Lemma~\ref{krylov matrix} can be extended to the case where $D$ has complex eigenvalues. This extension can relax the assumption in Assumption~\ref{assumption1} regarding the eigenvalues of matrix $A$ (the system poles), allowing for complex values. Specifically, we assume that the complex eigenvalues of $D$ lie within a region $\mathcal{D}$. In this scenario, the matrix $X_{n,mp}$ still satisfies the Sylvester matrix equation~\eqref{sylvester matrix equation}, preserving a displacement structure. Applying the results from~\cite{beckermann2017singular}, Inequality~\eqref{krylov matrix_result} remains valid, though the constant $\rho$ now depends on the region $\mathcal{D}$ rather than the fixed value $e^{\frac{\pi^2}{4}}$.
	Essentially, according to~\cite{beckermann2017singular}, one would need to construct a rational function of a certain degree that approximates $1$ on the complex region $\mathcal{D}$ and approximates $0$ on the set $\{z \in \mathbb{C} \mid z^{mp} = -1\}$. The constant $\rho$ would depend on the quality of such an approximation.}
	
	\ssun{In practice, there is no universal method to specify the region $\mathcal{D}$ for the complex poles, as this would require prior information about the system.}
\end{remark}

\begin{lemma}[\cite{beckermann2017singular}]\label{hankel matrix}
	Let $ H_n $ be an $ n \times n $ real positive definite Hankel matrix, then
	\begin{equation}
		\sigma_{j+2k}(H_n)\leq 16\rho^{-\frac{2k-2}{\log (2n)}}\sigma_{j}(H_n), 1 \leq j + 2k \leq n,
	\end{equation} 
	where $ \rho = e^{\frac{\pi^2}{4}} $.
\end{lemma}

\begin{proof}
	This is a direct corollary of Corollary 5.5 in the reference~\cite{beckermann2017singular}.
\end{proof}

\begin{theorem}[\cite{wedin1973perturbation}]\label{perturbation}
	\ssun{Given $\mathcal{A} \in \mathbb{C}^{m \times n}$, $\mathcal{B} \in \mathbb{C}^{m \times p}$ and $\mathcal{A}_1 = \mathcal{A} + \mathcal{E}_{\mathcal{A}}$, $\mathcal{B}_1 = \mathcal{B} + \mathcal{E}_{\mathcal{B}}$, suppose that the minimum Frobenius norm solutions to linear least squares problem $\min _x\| \mathcal{A} x - \mathcal{B} \|_{\rm F}$ and $\min _{{\widetilde{x}}}\| \mathcal{A}_1 {\widetilde{x}} - \mathcal{B}_1 \|_{\rm F}$ are $x^*$ and ${\widetilde{x}}^* = x^* + h$ respectively. If ${\rm rank}   [\mathcal{A}] = {\rm rank}  [\mathcal{A}_1] = n$ and $\|\mathcal{A}^\dagger\| \|\mathcal{E}_{\mathcal{A}}\|< 1$, then
		\begin{equation}\label{perturbation_1}
			\| h \| \leq \frac{\cond(\mathcal{A})}{\gamma_+ \|\mathcal{A}\|} \left( \|\mathcal{E}_{\mathcal{A}}\| \|x^*\| + \|\mathcal{E}_{\mathcal{B}}\|
			+ \frac{\cond(\mathcal{A})}{\gamma_+ \|\mathcal{A}\|} \|\mathcal{E}_{\mathcal{A}}\| \|r_x\|
			\right),
		\end{equation}
		where $\gamma_+ = 1 - \|\mathcal{A}^\dagger\| \|\mathcal{E}_{\mathcal{A}}\|$, and
		$r_x = \mathcal{B} - \mathcal{A}x^*$, $\eta = (\mathcal{A}^{\dagger})^{\rm H} x^*$.}
\end{theorem}

\begin{remark}
	\ssun{In fact, Theorem~\ref{perturbation} is the Frobenius norm version of Theorem 5.1 in~\cite{wedin1973perturbation}}. 
\end{remark}

\subsection{Proof of Lemma \ref{Ho-Kalman is ill-conditioned}}\label{OQH}
\begin{proof}
	Note that to make $ O $, $ Q $ full column and row rank respectively, the following inequalities must hold:
	\[
	n \leq mK_1 ,\, n \leq pK_2.
	\]
	Since $O,\,Q$	are order-$ n $ matrices, applying Lemma \ref{krylov matrix} gives that 
	\begin{equation}\label{ho-kalman-2}
		\sigma_{n}(O) \leq 4\rho^{-\frac{\left\lfloor \frac{n-1}{2m} \right\rfloor}{\log (2mK_1)} }\|O\| \leq 
		4\rho^{-\frac{\left\lfloor \frac{n-1}{2m} \right\rfloor}{\log (2mK_1)} }\|O\|_{\rm F}
	\end{equation}
	and
	\begin{equation}\label{ho-kalman-3}
		\sigma_{n}(Q) \leq 4\rho^{-\frac{\left\lfloor \frac{n-1}{2p} \right\rfloor}{\log (2pK_2)} }\|Q\| \leq 
		4\rho^{-\frac{\left\lfloor \frac{n-1}{2p} \right\rfloor}{\log (2pK_2)} }\|Q\|_{\rm F}.
	\end{equation}
	Then \eqref{O and Q} is a direct corollary of \eqref{ho-kalman-2} and \eqref{ho-kalman-3}.
	
	On the other hand, since $ H = OQ $, it is not difficult to get that
	\begin{equation}\label{ho-kalman-1}
		\sigma_{n}(H) \leq \min\left\{ \|Q\| \sigma_{n }(O), \|O\| \sigma_{n }(Q)\right\}
	\end{equation}
	by using the properties of singular values. It can be easily verified that 
	\begin{equation}\label{ho-kalman-4}
		\|O\|_{\rm F}^2 = \sum_{k=0}^{K_1-1}\|CA^k\|_{\rm F}^2 \leq \sum_{k=0}^{K_1-1}\|C\|_{\rm F}^2 \|A^k\|^2 \leq \overline{c}^2 mnK_1,
	\end{equation}
	where $ \overline{c} = \max_{i,j} |c_{ij}| $, 
	and
	\begin{equation}\label{ho-kalman-5}
		\|Q\|_{\rm F}^2 = \sum_{k=0}^{K_2-1}\|A^kB\|_{\rm F}^2 \leq \sum_{k=0}^{K_2-1}\|B\|_{\rm F}^2 \|A^k\|^2 \leq \overline{b}^2 pnK_2,
	\end{equation}
	where $ \overline{b} = \max_{i,j} |b_{ij}| $.
	Combining the above results and this fact that $ K_1 + K_2 = K $, we can complete the proof of \eqref{L}. 
\end{proof}

\subsection{Proof of Lemma~\ref{singular value of V}}

\begin{proof}
	Consider the following matrix $ \widetilde{V} $
	\begin{equation}\label{tilede{V}}
		\widetilde{V} = \begin{bmatrix}
			\widetilde{V}_{11}^\top & \widetilde{V}_{12}^\top & \cdots & \widetilde{V}_{mp}^\top
		\end{bmatrix}^\top_{pm(K-1) \times n},
	\end{equation}
	where for any $ 1 \leq i \leq m $, $ 1 \leq j \leq p $, 
	\begin{equation}
		\begin{aligned}
			\widetilde{V}_{ij} &= \begin{bmatrix}
				c_{i1}b_{1j}  & \cdots & c_{in}b_{nj} \\
				2c_{i1}b_{1j}\lambda_1 & \cdots & 2c_{in}b_{nj}\lambda_n \\
				\vdots & \ddots & \vdots \\
				(K-1)c_{i1}b_{1j}\lambda_1^{K-2} & \cdots & (K-1)c_{in}b_{nj}\lambda_n^{K-2} \\
			\end{bmatrix} \\
			&= \underbrace{\begin{bmatrix}
					1 & \\
					& 2 & \\
					& & \ddots & \\
					& & & K-1
			\end{bmatrix}}_{\Lambda_{K-1}} 
			\underbrace{\begin{bmatrix}
					1 & 1 & \cdots & 1 \\
					\lambda_1 & \lambda_2 & \cdots & \lambda_n \\
					\vdots & \vdots & \ddots & \vdots \\
					\lambda_1^{K-2} & \lambda_2^{K-2} & \cdots & \lambda_n^{K-2}
			\end{bmatrix}}_{\mathbf{V}_{K-1}} \\
			&\qquad \qquad\underbrace{\begin{bmatrix}
					c_{i1}b_{1j} & \\
					& c_{i2}b_{2j} & \\
					& &  \ddots & \\
					& & & c_{in}b_{nj}
			\end{bmatrix}}_{\Lambda_{ij}}.
		\end{aligned}
	\end{equation}
	It can be easily verified that $ V $ and $ \widetilde{V} $ have the same non-zero singular values, thus we only need to analyze the matrix $ \widetilde{V} $.
	Note that $ \widetilde{V} $ can be rewritten as 
	\begin{equation}\label{17}
		\widetilde{V}  = 
		\left( I_{pm} \otimes \Lambda_{K-1} \right) 
		\left( I_{pm} \otimes \mathbf{V}_{K-1} \right)
		\underbrace{
			\begin{bmatrix}
				\Lambda_{11}^\top &
				\Lambda_{12}^\top &
				\cdots &
				\Lambda_{pm}^\top
			\end{bmatrix}^\top}_{\Lambda}.
	\end{equation}
	Note that to make the matrix $V$ full rank, there must be that $K \geq \lceil \bm{\kappa} \rceil + 1 $, where $ \bm{\kappa} = n/(pm) $, then it follows that 
	\begin{equation}\label{overline V}
		\sigma_{\min}^2(V) = 
		\sigma_{\min}^2(\widetilde{V}) \leq K^2  \|\Lambda\|^2 \lambda_{\lfloor \bm{\kappa} \rfloor}(\mathbf{V}_{K-1}\mathbf{V}_{K-1}^{\top}).
	\end{equation}
	Note that $ \| \Lambda \|^2 \leq \| \Lambda \|_{\rm F}^2  \leq  n pm \overline{\bm{\delta}}^2 $, 
	thus next we only analyze $ \lambda_{\lfloor \bm{\kappa} \rfloor}(\mathbf{V}_{K-1}\mathbf{V}_{K-1}^{\top})$.
	Using Lemma \ref{hankel matrix} above yields that 
	\begin{equation}\label{V_{K-1}}
		\lambda_{\lfloor \bm{\kappa} \rfloor}(\mathbf{V}_{K-1}\mathbf{V}_{K-1}^{\top}) \leq 16 \rho^{-\frac{\lfloor \bm{\kappa} \rfloor -3}{\log (2K)}}  \| \mathbf{V}_{K-1}\mathbf{V}_{K-1}^{\top}\|.
	\end{equation}
	Note that $ \| \mathbf{V}_{K-1}\mathbf{V}_{K-1}^{\top}\| \leq \|\mathbf{V}_{K-1}\|_{\rm F}^2 \leq nK $, then combining it with the above results can complete the proof.
\end{proof}

\subsection{Proof of Theorem~\ref{moesp}}

\begin{proof}
	\ssun{Note that $\overline{A} = \underline{O}^\dagger \overline{O}$ can be expressed as the equivalent least-squares problem $ \min_{A} \left\|   \underline{O} A - \overline{O}
		\right\|_{\rm F}$.
		By applying the results of Theorem~\ref{perturbation} to the original least squares problem and its perturbed version, and observing that $ \| \Delta_{\underline{O}} \|  = \| \underline{\widehat{O}} -  \underline{O} \| \leq  \| \Delta_{O}\|$ and $\| \Delta_{\overline{O}} \|  = \| \overline{\widehat{O}} -  \overline{O} \| \leq  \| \Delta_{O}\|$, the proof follows.}
\end{proof}

\subsection{Proof of Theorem~\ref{fisher information matrix}}
\begin{proof}
	The dynamic response of the system \eqref{linear_system} in the first $ K $ time steps is as follow
	\begin{equation}\notag
		\underbrace{
			\begin{bmatrix}
				y_1 \\ y_2 \\ \vdots \\ y_{K}
		\end{bmatrix}}_{Y_{[1,K]}} = 
		\underbrace{
			\begin{bmatrix}
				H_0^{} &  &  &   \\
				H_1^{} & H_0^{} &  &  \\
				\vdots & \vdots & \ddots &  \\
				H_{K-1}^{} & H_{K-2}^{} & \cdots & H_0^{}
		\end{bmatrix}}_{\mathcal{H}_k}
		\underbrace{
			\begin{bmatrix}
				u_0^{} \\ u_1^{} \\ \vdots \\ u_{K-1}^{}
		\end{bmatrix}}_{U_{[0,K-1]}} 
		+ 		\underbrace{
			\begin{bmatrix}
				v_1 \\ v_2 \\ \vdots \\ v_{K}
		\end{bmatrix}}_{V_{[1,K]}},
	\end{equation}
	where $ H_k  = CA^kB$ is the $ k$-th Markov parameter of the system \eqref{linear_system}.

	Since the measurement noise $ V_{[1,K]} $ is Gaussian, and the input $ U_{[0,K-1]} $ is known, the output $ Y_{[1,K]} $ is also Gaussian and its probability density function conditioned on the value of system poles $ \bm{\lambda}$ can be obtained as
	\begin{equation}\label{pdf}
		f_{\bm{\lambda}}\left(Y_{[1,K]}  \right) = \frac{\exp\left(-\frac{1}{2} \left(Y_{[1,K]}  - \mu_K\right)^\top  \left(Y_{[1,K]}  - \mu_K\right) 	\right)}{\sqrt{(2\pi)^K }},
	\end{equation}
	where $ \mu_K \triangleq \mathcal{H}_K U_{[0,K-1]} $ represents the mean of $ Y_{[1,K]} $.
	According to the definition of  Fisher Information matrix, the $(i,j) $ element of the Fisher Information matrix $ \mathbf{I}\left(\bm{\lambda} \right) $ of unknown poles $ \bm{\lambda} $ has the following form
	\begin{equation}\notag
		\mathbf{I}_{i,j}\left(\bm{\lambda} \right) = \mathbb{E} \left[
		\frac{\partial}{\partial \lambda_i}\log f_{\bm{\lambda}}\left(Y_{[1,K]}\right)
		\frac{\partial}{\partial \lambda_j}\log f_{\bm{\lambda}}\left(Y_{[1,K]}\right)
		\right],
	\end{equation}
	where the notation $ \mathbb{E} $ is the expectation operator with respect to $ f_{\bm{\lambda}}\left(Y_{[1,K]}\right) $.
	Since the output $ Y_{[1,K]} $ is Gaussian, based on the results of section 3.9 of \cite{kay1993fundamentals}, we can further get that 
	\begin{equation}\label{fisher_information_31}
		\begin{aligned}
			\mathbf{I}_{i,j}\left(\bm{\lambda} \right) &= \left[\frac{\partial \mu_K}{\partial \lambda_i}\right]^\top \frac{\partial \mu_K}{\partial \lambda_j},
		\end{aligned}
	\end{equation}
	where $ \frac{\partial \mu_K}{\partial \lambda_i} = \begin{bmatrix}
		\frac{\partial [\mu_K]_1}{\partial \lambda_i} &
		\frac{\partial [\mu_K]_2}{\partial \lambda_i} &
		\cdots &
		\frac{\partial [\mu_K]_{mK}}{\partial \lambda_i}
	\end{bmatrix}^\top  $, and $ [\mu_K]_j $ denotes the $ j $-th entry of $ \mu_K $ for each $ j = 1,2,\cdots,mK $.
	Based on \eqref{fisher_information_31}, the Fisher Information matrix $ \mathbf{I}(\bm{\lambda}) $ can be obtained as follow
	\begin{equation}\label{fisher_information_A}
		\mathbf{I}\left(\bm{\lambda} \right) = 
		\underbrace{\begin{bmatrix}
				\frac{\partial \mu_K}{\partial \lambda_1} &
				\cdots & 
				\frac{\partial \mu_K}{\partial \lambda_{n}}
			\end{bmatrix}^\top}_{\mathcal{L}_{\bm{\lambda}}^\top}
		\underbrace{\begin{bmatrix}
				\frac{\partial \mu_K}{\partial \lambda_1} &
				\cdots & 
				\frac{\partial \mu_K}{\partial \lambda_{n}}
		\end{bmatrix}}_{\mathcal{L}_{\bm{\lambda}}}.
	\end{equation}	
	According to the chain derivation rule, $ \mathcal{L}_{\bm{\lambda}} $ can be decomposed as follow
	\begin{equation}\label{L=SV}
		\begin{aligned}
				&\mathcal{L}_{\bm{\lambda}} = \begin{bmatrix}
				\frac{\partial \mu_K}{\partial \lambda_1} &
				\cdots & 
				\frac{\partial \mu_K}{\partial \lambda_{n}}
			\end{bmatrix} = \\ 
			&\underbrace{
				\begin{bmatrix}
					\frac{\partial \mu_{K}}{\partial 		[H_{0}]_{11}} &  \cdots & \frac{\partial \mu_{K}}{\partial [H_{K-1}]_{pm}}
			\end{bmatrix}}_{S}
			\underbrace{
				\begin{bmatrix}
					\frac{\partial [H_{0}]_{11}}{\partial \lambda_{1}} &  \cdots & \frac{\partial  [H_{0}]_{11}}{\partial \lambda_{n}} \\ 
					\frac{\partial  [H_{0}]_{12}}{\partial \lambda_{1}} &  \cdots & \frac{\partial  [H_{0}]_{12}}{\partial \lambda_{n}} \\
					\vdots  & \ddots & \vdots \\
					\frac{\partial  [H_{K-1}]_{pm}}{\partial \lambda_{1}} &  \cdots & \frac{\partial  [H_{K-1}]_{pm}}{\partial \lambda_{n}}
			\end{bmatrix}}_{V}.
		\end{aligned}
	\end{equation}
	Combined with the expressions of $ \mu_{K} $ and $ H_k $, it is not difficult to get the expressions of $ S $  and $V$ as shown in \eqref{S} and \eqref{V}, respectively.
\end{proof}

\subsection{Proof of Theorem \ref{worst case-cramer-rao}}
\begin{proof}
	According to Theorem \ref{worst case}, Theorem \ref{cramer-rao bound} and the proof of Theorem \ref{fisher information matrix}, we only need to verify that the probability density function $ f_{\bm{\lambda}}\left(Y_{[1,K]}  \right) $ in \eqref{pdf} satisfies the regularity conditions
	\[
	\mathbb{E}\left[\frac{\partial \log f_{\bm{\lambda}}\left(Y_{[1,K]}  \right)}{\partial \lambda_i} \right] = 0, \text{ for each } i = 1,2,\cdots,n,
	\]
	where the expectation is taken with respect to $ f_{\bm{\lambda}}\left(Y_{[1,K]}  \right) $. 
	For each $ i = 1,2,\cdots,n $, it follows that
	\begin{equation}\label{key}
		\begin{aligned}
			\frac{\partial \log f_{\bm{\lambda}}\left(Y_{[1,K]}  \right)}{\partial \lambda_i} &=
			-\frac{1}{2}\frac{\partial}{\partial \lambda_i}\left(Y_{[1,K]}  - \mu_K\right)^\top  \left(Y_{[1,K]}  - \mu_K\right) \\
			&=
			\sum_{j=1}^{mK} \left(\left[Y_{[1,K]}\right]_j - \left[\mu_{K}\right]_j\right) \frac{\partial}{\partial \lambda_i}\left[\mu_{K}\right]_j,
		\end{aligned}
	\end{equation}
	where $ \left[Y_{[1,K]}\right]_j $ and $ [\mu_K]_j $ denote the $ j $-th entry of $ \mu_K $ and $ Y_{[1,K]} $, respectively. Now take the expectation of both sides with respect to $ f_{\bm{\lambda}}\left(Y_{[1,K]}  \right) $, then we can get that
	\begin{multline}
		\mathbb{E}\left[\frac{\partial \log f_{\bm{\lambda}}\left(Y_{[1,K]}  \right)}{\partial \lambda_i} \right] \\
		=\mathbb{E} \left[ 
		\sum_{j=1}^{mK} \left(\left[Y_{[1,K]}\right]_j - \left[\mu_{K}\right]_j\right) \frac{\partial}{\partial \lambda_i}\left[\mu_{K}\right]_j
		\right]
		\\ 
		= \sum_{j=1}^{mK} \mathbb{E} \left[\left[Y_{[1,K]}\right]_j - \left[\mu_{K}\right]_j\right] \frac{\partial}{\partial \lambda_i}\left[\mu_{K}\right]_j 
		= 0,
	\end{multline}
	where the last equality holds because $ Y_{[1,K]} $ is Gaussian with mean $ \mu_{K} $, which completes the proof.
\end{proof}

\end{document}